\documentclass[a4paper,UKenglish,cleveref, autoref, thm-restate]{lipics-v2021}
\pdfoutput=1 

\usepackage{mathtools}
\usepackage{stmaryrd}
\usepackage{subcaption}

\usepackage{tikz}
\usetikzlibrary{automata,positioning,arrows.meta,arrows,decorations.pathmorphing,backgrounds,calc}

\usepackage[textwidth=1.9cm,textsize=small]{todonotes}

\newcommand{\N}{\mathbb{N}}
\newcommand{\Z}{\mathbb{Z}}
\newcommand{\Q}{\mathbb{Q}}

\newcommand{\A}{\mathcal{A}}
\newcommand{\B}{\mathcal{B}}
\newcommand{\W}{\mathcal{W}}
\newcommand{\C}{\mathcal{C}}

\newcommand{\bF}{\mathbb{F}}

\newcommand{\set}[1]{\{#1\}}
\renewcommand{\vec}[1]{{\bf #1}}

\DeclareMathOperator{\PSF}{PSF}

\DeclareMathOperator{\chr}{char}

\DeclareMathOperator{\Poly}{Poly}
\DeclarePairedDelimiter{\card}{|}{|}

\title{Pumping-Like Results for Copyless Cost Register Automata and Polynomially Ambiguous Weighted Automata}

\titlerunning{Pumping-Like Results for CCRA and Polynomially Ambiguous WA}

\author{Filip Mazowiecki}{University of Warsaw, Poland}{f.mazowiecki@mimuw.edu.pl}{}{Supported by Polish National Science Centre
SONATA BIS-12 grant number 2022/46/E/ST6/00230}

\author{Antoni Puch}{University of Warsaw, Poland}{a.puch@student.uw.edu.pl}{}{Supported by Polish National Science Centre
SONATA BIS-12 grant number 2022/46/E/ST6/00230}

\author{Daniel Smertnig}{University of Ljubljana and Institute of Mathematics, Physics, and Mechanics (IMFM), Slovenia}{daniel.smertnig@fmf.uni-lj.si}{https://orcid.org/0000-0002-5391-2471}{Supported by the Slovenian Research and Innovation Agency (ARIS) program P1-0288 and grant J1-60025}

\authorrunning{F. Mazowiecki, A. Puch and D. Smertnig}

\Copyright{Filip Mazowiecki, Antoni Puch and Daniel Smertnig} 



\category{} 

\relatedversion{} 

\nolinenumbers 

\EventEditors{John Q. Open and Joan R. Access}
\EventNoEds{2}
\EventLongTitle{42nd Conference on Very Important Topics (CVIT 2016)}
\EventShortTitle{CVIT 2016}
\EventAcronym{CVIT}
\EventYear{2016}
\EventDate{December 24--27, 2016}
\EventLocation{Little Whinging, United Kingdom}
\EventLogo{}
\SeriesVolume{42}
\ArticleNo{23}

\begin{document}

\maketitle

\begin{abstract}
In this work we consider two rich subclasses of weighted automata over fields: polynomially ambiguous weighted automata and copyless cost register automata. Primarily we are interested in understanding their expressiveness power. Over the field of rationals and $1$-letter alphabets, it is known that the two classes coincide; they are equivalent to linear recurrence sequences (LRS) whose exponential bases are roots of rationals. We develop a tool we call Pumping Sequence Families, which, by exploiting the simple single-letter behaviour of the models, yields two pumping-like results over arbitrary fields with unrestricted alphabets, one for each class. As a corollary of these results, we present examples proving that the two classes become incomparable over the field of rationals with unrestricted alphabets.

We complement the results by analysing the zeroness and equivalence problems. For weighted automata (even unrestricted) these problems are well understood: there are polynomial time, and even NC$^2$ algorithms. For copyless cost register automata we show that the two problems are \textsc{PSpace}-complete, where the difficulty is to show the lower bound.
\end{abstract}

\clearpage
\section{Introduction}
Weighted automata are a computational model assigning values from a fixed domain to words~\cite{droste2009handbook}. The domain can be anything with a semiring structure. Typical examples are: fields~\cite{Schutzenberger61b}, where in particular probabilistic automata assign to every word the probability of its acceptance~\cite{paz71}; and tropical semirings, popular due to their connection with star height problems~\cite{Hashiguchi88}. In this paper we focus on weighted automata over fields. These are finite automata with transitions, input and output edges additionally labeled by weights from the field. On an input word the value of a single run is the product of all weights, and the output of the weighted automaton is the sum of values over all runs. See \cref{fig:examples} for examples.

Unlike finite automata, nondeterminism makes weighted automata more expressive. This naturally leads to the decision problem of \emph{determinisation}: given a weighted automaton does there exist an equivalent deterministic one? Over fields, it was recently shown that the problem is decidable~\cite{BellS23}, later improved to a 2-\textsc{ExpTime} upper bound on the running time~\cite{BenaliouaLR24}. Both papers rely on techniques used to obtain Bell and Smertnig's result~\cite{bell2021noncommutative} characterising the intermediate class of \emph{unambiguous weighted automata}: a subclass that allows nondeterminism, but for every word there is at most one run of nonzero value. The authors proved Reutenauer's conjecture~\cite{Reutenauer79}, which we explain below.

\begin{figure}
\centering
\begin{subfigure}[t]{.50\textwidth}
\centering
\begin{tikzpicture}
\node[state] (p) {$p_1$};
\node[state,right=1cm of p] (p2) {$p_2$};

\node[state,right = 0.9cm of p2] (q) {$q_1$};
\node[state,right = 1cm of q] (q2) {$q_2$};

\path
(p) edge[->,bend left] node[above] {$a \mid 2$} (p2)
(p2) edge[->,bend left] node[below] {$a \mid 2$} (p)
(p2.north) edge[->] node[right] {$1$} ++ (0,0.5)
(q) edge[->,bend left] node[above] {$a \mid 3$} (q2)
(q2) edge[->,bend left] node[below] {$a \mid 3$} (q)
(q.north) edge[->] node[right] {$1$} ++ (0,0.5)
(p.west) edge[<-] node[above] {$1$} ++ (-0.5,0)
(q.west) edge[<-] node[above] {$1$} ++ (-0.5,0)
;
\end{tikzpicture}
\caption{A four state unambiguous weighted automaton $\A$ over a $1$-letter alphabet $\set{a}$. Nonzero initial labels for $p_1$ and $q_1$ have value $1$. The nonzero final states are $p_2$ and $q_1$, both with weight $1$. For every word $a^n$ there are two runs: on the left of value $1 \cdot 2^n$; on the right of value $1 \cdot 3^n$. Depending on the parity of $n$ only one of the runs has nonzero output, thus $\A(a^n) = 3^n$ for even $n$, and $\A(a^n) = 2^n$, otherwise.}\label{fig:unambiguous}
\end{subfigure}
\hspace*{\fill}
\begin{subfigure}[t]{.45\textwidth}
\centering
\begin{tikzpicture}
\node[state] (p) {};
\node[state,right=1.1cm of p] (q) {};

\path
(p) edge[->,loop above] node[above] {$a \mid 1$} (p)
(q) edge[->,loop above] node[above] {$a \mid 1$} (q)
(p) edge[->] node[above] {$a \mid 1$} (q)
(p.west) edge[<-] node[above] {$1$} ++ (-0.5,0)
(q.east) edge[->] node[above] {$1$} ++ (0.5,0)
;
\end{tikzpicture}
\caption{A 2-state polynomially ambiguous weighted automaton $\B$ over the alphabet $\set{a}$. Note that every word $a^n$ has $n$ runs with value $1$, hence $\B(a^n) = n$.}
	\label{fig:n}
\end{subfigure}
\caption{Weighted automata over the field of rationals $\Q(+,\cdot)$. For clarity, we omit zero labels.}\label{fig:examples}
\end{figure}
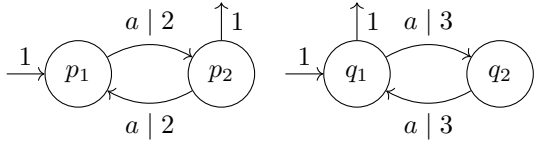
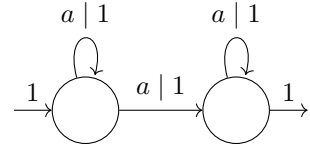

Given a weighted automaton $\W$ over the alphabet $\Sigma$ consider $\W(\Sigma^*)$, the set of all outputs over all words. For example, in \cref{fig:examples} we have: $\A(a^*) = \set{\,2^{2n+1} \mid n \in \N\,} \cup \set{\, 3^{2n} \mid n \in \N\,}$; and $\B(a^*) = \N$. Reutenauer's conjecture (now Bell and Smertnig's Theorem) states that for every weighted automaton $\W$ over a field $K$: there exists an equivalent unambiguous weighted automaton if and only if there exists a finitely generated multiplicative subgroup $G \subseteq K$ such that $\W(\Sigma^*) \subseteq G \cup\{0\}$. For example $\A(a^*) \subseteq G_{\A}$, where $G_{\A}$ is generated by the transition weights $\set{2,3}$. It is not hard to see that this construction generalises to every unambiguous weighted automaton, the crux is to prove the other implication. As an immediate nontrivial application, notice that there is no unambiguous weighted automaton equivalent to $\B$, as the set $\N \setminus \set{0}$ is not contained in any finitely generated subgroup.

The equivalence problem for weighted automata over fields is famously decidable in polynomial time~\cite{Schutzenberger61b}. However, most natural problems are undecidable~\cite{paz71,FijalkowRW22,DaviaudJLMP021,CzerwinskiLMPW22}. This triggered the study of intermediate classes between deterministic and unrestricted weighted automata. One way to define such a class is based on \emph{ambiguity}, generalising unambiguous weighted automata. A much broader class are \emph{polynomially ambiguous weighted automata}, where the number of accepting runs is bounded by a polynomial in the size of the input word (see \cref{fig:n}). Restricting the input automaton to polynomially ambiguous can significantly lower the complexity of a problem, for example, the discussed problem of determinisation is known to be in \textsc{PSpace} over the field of rationals~\cite{JeckerMP24}. Another way to define such a class comes from \emph{cost register automata}~\cite{AlurDDRY13} (CRA), a deterministic model with polynomial register updates. In this context it is natural to consider its copyless restriction (CCRA), because every function recognisable by a CCRA is also recognisable by a weighted automaton~\cite{MazowieckiR15} (for simplicity the definition of CCRA is postponed to \cref{sec:preliminaries}).

As far as we know, these two classes, polynomially ambiguous weighted automata and copyless CRA, are the richest studied classes that are known to be strictly contained in the class of unrestricted weighted automata. Over the tropical semiring they are known to be incomparable in terms of expressiveness~\cite{MazowieckiR19,ChattopadhyayMM21}, which suggests the same over fields. One attempt to prove this result was in~\cite{BarloyFLM22}, where the authors considered weighted automata over the field of rationals with $1$-letter alphabets. By identifying words $a^n$ with their length $n$, one can view such automata as sequences. In fact weighted automata over $1$-letter alphabets are equivalent to the well-known class of \emph{linear recurrence sequences} (LRS)~\cite{OuaknineW15}. In~\cite{BarloyFLM22}, the authors prove that polynomially ambiguous weighted automata and copyless CRA coincide, and that they are also equivalent to the class of LRS whose exponential bases are roots of rationals. This means that in the exponential polynomial representation of LRS: $\sum_{i=1}^n p_i(x) \lambda_i^x$, for every $i$ there is an $n_i$ such that $\lambda_i^{n_i} \in \Q$. In particular this shows that the Fibonacci sequence does not belong to this class, as the golden ratio $\varphi$ is not of this form.

\paragraph*{Our Contribution}
Our work can be seen as a follow-up to~\cite{BarloyFLM22}. An immediate corollary of our results is that polynomially ambiguous weighted automata and copyless CRA over the field of rationals are incomparable classes in terms of expressiveness. To prove this, we developed a tool we call Pumping Sequence Families ($\PSF$), which allows us to exploit the behaviour of these classes over $1$-letter alphabets. In the following we use the standard sequence notation $(a_n)_n = a_0, a_1, a_2,\ldots$.

\begin{definition} \label{d:psf}
    A Pumping Sequence Family of a function $f: \Sigma^* \mapsto A$, for any set $A$, is the set of all sequences of the form $\hat{f}(u, w, v) \coloneqq (f(uw^nv))_n$, with $u, w, v$ ranging over all words in $\Sigma^*$. We denote it by $\PSF(f)$.
\end{definition}
One should think that $f$ is being projected onto many single-letter-like cases at once, where $w$ plays the role of the single letter in the alphabet, while $u$ and $v$ correspond to slight adjustments of respectively the initial and acceptance conditions. The definition of $\PSF(f)$ exploits that $u$, $v$ and $w$ range over all words, which captures behaviour beyond the single letter alphabet case. Using fixed words one cannot differentiate polynomially ambiguous weighted automata and copyless CRA due to~\cite[Theorem 6, Theorem 13]{BarloyFLM22}. However, this simple extension of the single letter case analysis will be enough to show those models to be incomparable. To showcase our approach, let us consider a simple application for languages, where $A = \set{\top, \bot}$ (meaning acceptance and rejection of the word).

\begin{example}\label{example:PSFfinite}
Consider $f : \set{a}^* \mapsto \set{\top, \bot}$ which maps all words of even length to $\top$ and all others to $\bot$. Then $\PSF(f)$ consists of four sequences (for simplicity we write an example generator for each sequence):
\begin{itemize}
\item $\hat{f}(\varepsilon, \varepsilon, \varepsilon) = \top, \top, \top, \top,  \dots$
\item $\hat{f}(a, \varepsilon, \varepsilon) = \bot, \bot, \bot, \bot, \dots$
\item $\hat{f}(\varepsilon, a, \varepsilon) = \top, \bot, \top, \bot, \dots$
\item $\hat{f}(a, a, \varepsilon) = \bot, \top, \bot, \top, \dots$.
\end{itemize}
\end{example}

\begin{example}\label{example:PSFinfinite}
Consider $g : \set{a, b}^* \mapsto \set{\top, \bot}$ which maps all words of the form $a^nb^n$ to $\top$ and all others to $\bot$. Then $\PSF(g)$ is an infinite set:
\begin{itemize}
\item $\hat{g}(ab, \varepsilon, \varepsilon) = \top, \top, \top, \top, \dots$
\item $\hat{g}(a, \varepsilon, \varepsilon) = \bot, \bot, \bot, \bot, \dots$
\item $\hat{g}(\varepsilon, ab, \varepsilon) = \top, \top, \bot, \bot, \dots$
\item $\hat{g}(a^k, b, \varepsilon) = \underbrace{\bot,\ldots \bot}_{k}, \top, \bot, \bot, \dots$ for every $k\ge 0$.
\end{itemize}
\end{example}


The above examples already present us with a simple use case for Pumping Sequence Families. By looking at the transition function of the underlying DFA we can see that, generalising \cref{example:PSFfinite}, for a regular language the Pumping Sequence Family of its characteristic function will be finite. However, as witnessed in \cref{example:PSFinfinite}, the Pumping Sequence Family of the context-free language $L = \set{\,a^kb^k \mid k\in \N\,}$ is infinite, proving that it is not regular and thus differentiating regular and context-free languages. Note that, due to Parikh's theorem, over $1$-letter alphabets the two language classes are equivalent and semi-linear. Meaning if we would fix $u$, $v$ and $w$ then simply looking at the single letter case behaviour we would not be able to differentiate these classes.

We now present how we use Pumping Sequence Families for weighted functions. There the set $A$ from \cref{d:psf} is simply the underlying field.

For a function $h$ recognised by a Copyless Cost Register Automaton we will restrict elements of its Pumping Sequence Family. Note that, since copyless CRA are a subset of weighted automata, elements of $\PSF(h)$ can be essentially represented as exponential polynomials. Consider such a sequence $a_n = \sum_{i=1}^d p_i(n) \lambda_i^n$. Let us write the polynomials $p_i$ explicitly: $p_i(x) = \sum_{j=1}^{m_i} \alpha_{i,j} x^j$. For every degree $k$ we define the sum of $k$-degree coefficients $S_k((a_n)_n) = \sum_{i=1}^d\alpha_{i,k}$. In \cref{theorem:CCRA} we show that, up to minor technical details, for all $k$ the set $\set{\,S_k((a_{n})_n) \; | \; (a_n)_n \in \PSF(h) \,}$ is contained in a finitely generated subsemiring $R$. For intuition, if we consider the generators $\{\tfrac{1}{2}, \tfrac{1}{3}\}$, by adding, subtracting and multiplying, they generate $R = \set{\,\tfrac{a}{6^k} \mid a \in \Z, k \in \N\,}$. This allows us to give an example of a polynomially ambiguous automaton that is not definable by any copyless CRA (the proof is short but technical, see \cref{exm:wfa-not-ccra}).

For polynomially ambiguous weighted automata, our work is inspired by~\cite{puch-smertnig24}, where the authors attempt to characterise polynomially ambiguous weighted automata in a similar manner to Bell and Smertnig's Theorem. For a function recognised by a polynomially ambiguous weighted automaton $\W$ and given a sequence $(c_n)_n \in \PSF(\W)$ consider again its exponential polynomial $\sum_{i=1}^d p_i(x) \lambda_i^x$ and let $E((c_n)_n) =\set{\,\lambda_i \mid 1 \le i \le d\,}$ be the set of exponential bases. In \cref{t:poly_fin_gen} we show that the set $\bigcup_{(c_n)_n \in \PSF(\W)} E((c_n)_n)$ is contained in a finitely generated subgroup $G$. We provide a self-contained proof, and we show that our property, which is simpler to work with when considering examples, is equivalent to the one in~\cite{puch-smertnig24} (conjectured to characterise polynomially ambiguous automata). We obtain a corresponding, simple but technical, example of a copyless CRA that is not definable by any polynomially ambiguous weighted automaton (\cref{exm:ccra-not-pa}).

In this context a natural question is whether our property for copyless CRA can be a characterisation. We conjecture that it is not the case and that, in some sense, such a characterisation should not exist. We show examples of functions that satisfy the property we developed for CCRA, but we find it unlikely that there are CCRA that define them. More generally, in~\cite{MazowieckiR19} the authors prove that the class of CCRA is not closed under reversal for the tropical semiring. More precisely, there is a CCRA $\C$ such that there is no CCRA $\C'(w) = \C(w^r)$, where $w^r$ is $w$ reversed. We conjecture that over fields CCRA are also not closed under reversal, which makes such characterisations unlikely.

Our final contribution is the analysis of the equivalence and zeroness problems for both classes. As already mentioned for weighted automata (even without restrictions) equivalence and zeroness are in polynomial time~\cite{Schutzenberger61b} and even in NC$^2$~\cite{Tzeng96}. For copyless CRA the translation to weighted automata~\cite{MazowieckiR15} yields an exponential blow up in the size of the automaton (we provide a self-contained short translation). Since problems in NC$^2$ can be solved sequentially in polylogarithmic space~\cite{Ruzzo81}, this yields a trivial \textsc{PSpace} algorithm. Our contribution is a matching \textsc{PSpace} lower-bound.

\begin{restatable}[]{theorem}{theorempspace}
\label{theorme:pspace}
Zeroness and equivalence problems are \textsc{PSpace}-complete for CCRA over $\Q$.
\end{restatable}

\paragraph*{Organisation}
We start with definitions in \cref{sec:preliminaries}. In \cref{sec:ccra} and \cref{sec:polyamb} we prove the properties of copyless CRA, and polynomially ambiguous weighted automata, respectively; and we present examples separating the classes. In Appendix~\ref{sec:equivalence} we discuss the decision problems.

\section{Preliminaries} \label{sec:preliminaries}
Let $\N\coloneqq \{0,1,2, \dots\}$.
For a field $K$, let $K^\times \coloneqq K\setminus \{0\}$ denote the multiplicative group of nonzero elements.
We sometimes write $1_K$ and $0_K$ for the elements $1$ and $0$ of the field, to emphasize which $1$ and $0$ we mean.

\subsection{Automata and Sequences}

A \emph{weighted automaton} over a field $K$ is a tuple $\A = (d, \Sigma, (M(a))_{a \in \Sigma}, I, F)$, where: $d \in \N$ is its dimension; $\Sigma$ is a finite alphabet; $M(a) \in K^{d \times d}$ are transition matrices; $I$,~$F \in K^d$ are the initial and final vectors, respectively. 
For simplicity, sometimes we will write $\A = (d, M, I, F)$, that is, we will omit $\Sigma$ in the tuple.

Weighted automata can be defined more generally over semirings, but in this paper we only consider the case of fields.
The field $\Q$ is already rich enough to
produce all phenomena of interest to us.
Thus, our examples will be for $\Q$ with the usual addition and product, unless stated otherwise.

Given a word $w = w_1\ldots w_n \in \Sigma^*$, we denote $M(w) \coloneqq M(w_1) \cdot\ldots\cdot M(w_n)$.
In particular $M(\epsilon)$ is the $d\times d$-identity matrix.
A weighted automaton defines a function $\A\colon \Sigma^* \to K$, by $\A(w) \coloneqq I^\intercal \cdot M(w) \cdot F$.
We say that a weighted automaton $\A$ is a \emph{linear recurrence sequence \textup(LRS\textup)} if $|\Sigma| = 1$. 
Then, by identifying $\Sigma^*$ with $\N$, that is, identifying the word $a^n$ with its length $n$, we write that $\A\colon \N \to K$.
We will also denote such sequences $(a_n)_n$ instead of $\A$, where $a_n \coloneqq \A(n)$.

\begin{example}\label{example:linear}
Consider an LRS $\A = (2, M, I, F)$, where: $M = \begin{psmallmatrix} 1 & 1  \\ 0 & 1 \end{psmallmatrix}$; $I = \begin{psmallmatrix} 1 \\ 0 \end{psmallmatrix}$ and $F = \begin{psmallmatrix} 0 \\ 1 \end{psmallmatrix}$. Then $\A(n) = a_n = n$.
\end{example}

For weighted automata we define their underlying automata.
A weighted automaton $\A = (d, \Sigma, (M(a))_{a \in \Sigma}, I, F)$ can be interpreted as an automaton with states $\set{1,\ldots,d}$ such that for every $a \in \Sigma$ a nonzero entry in $M_a[i,j]$ defines a transition from $i$ to $j$ labeled by $a$ of weight $M_a[i,j]$ (thus we ignore transitions of weight $0$).
Similarly, initial and final states are $i$ such that $I[i]$ and $F[i]$ are nonzero, respectively.
Their weights are $I[i]$ and $F[i]$.
By ignoring the weights of transitions, initial and final states, we obtain a finite automaton $\B$, which we call the \emph{underlying automaton} of $\A$. 

\begin{example}\label{example:automaton_linear}
The LRS $\A = (2, M, I, F)$ in \cref{example:linear} is an equivalent presentation of the weighted automaton $\B$ in \cref{fig:n}.
\end{example}

A weighted automaton $\A$ is \emph{polynomially ambiguous} if there is a polynomial function $p\colon \N \to \N$ such that for every $w \in \Sigma^*$ the number of accepting runs of the underlying automaton on $w$ is bounded by $p(|w|)$.
For example, the automaton in \cref{example:automaton_linear} is polynomially ambiguous as it suffices to take $p(n) = n$.
In general this is a strict subclass: there exist weighted automata that are not equivalent to any polynomially ambiguous weighted automata.

LRS can be characterised in another way. 
An LRS $(a_n)_n$ can be defined by a (homogeneous) recurrence relation of the form $a_{n+k} = \sum_{i=0}^{k-1} c_i \cdot a_{n+i}$ with $c_i \in K$ and $k$ initial values $a_0$, \dots,~$a_{k-1}$.
Here $k$ is the \emph{order} of the recurrence.
For instance, the LRS $(a_n)_n$ from \cref{example:linear} can be defined by $a_{n+2} = 2a_{n+1} - a_n$ and $a_0 = 0$, $a_1 = 1$. 
It is well-known that the two definitions coincide~\cite[Lemma 1.1]{halava2005skolem} \cite[Proposition 2.1]{CadilhacMPPS24}.
Moreover, the translation is effective in polynomial time, and under this translation, the dimension $d$ of the weighted automaton equals the order $k$ of the recurrence.

Any given LRS $(a_n)_n$ satisfies many different linear recurrences.
However, it is well-known that there is a unique (homogeneous) recurrence of minimal order satisfied by $(a_n)_n$ \cite[Ch.~6.1]{berstel-reutenauer11}.
The corresponding order $k$ is then the \emph{order} of the LRS.
This minimal recurrence gives rise to the \emph{characteristic polynomial} $q=x^k - c_{k-1} x^{k-1} - \dots - c_0$ of $(a_n)_n$ \cite[Ch.~6.1]{berstel-reutenauer11} \cite{everest2003recurrence}.
The roots of the characteristic polynomial (considered in the algebraic closure $\overline K$) are the \emph{characteristic roots} of the LRS $(a_n)_n$.

\begin{example}
    Continuing from \cref{example:linear}, the characteristic polynomial is $q=x^2-2x+1 = (x-1)^2$.
    Hence, the only characteristic root is $1$ (with multiplicity $2$).
\end{example}

The characteristic roots of LRS definable by polynomially ambiguous weighted automata are always roots of elements of $K$ \cite{BarloyFLM22,Kostolanyi23,puch-smertnig24}.
So, for example, the Fibonacci sequence $F_0=0$, $F_1=1$, $F_{n+2}=F_{n+1}+F_n$, is not recognised by a polynomially ambiguous weighted automaton over $\Q$, as its characteristic roots are the golden ratio $\varphi=\tfrac{1+\sqrt5}{2}$ and $\psi=\tfrac{1-\sqrt{5}}{2}$.

We recall an additional characterisation of LRS, namely as coefficient sequences of rational functions, leading to exponential polynomials.
See also \cite[Chapter~6]{berstel-reutenauer11}\cite{everest2003recurrence}\cite[Proposition 2.11]{halava2005skolem} or Appendix~\ref{subsec:appendix-exppoly}.
A sequence $(a_n)_n$ is an LRS if and only if the (formal) generating series $F = \sum_{n=0}^\infty a_n x^n \in K\llbracket x \rrbracket$ is a rational function: the series $F$ is the formal Taylor series expansion of some $p/q$ with $p$,~$q \in K[x]$ coprime polynomials and $q(0)\ne 0$.
For nonzero $\lambda \in \overline{K}$, the following are now equivalent:
\begin{itemize}
\item $\lambda$ is a characteristic root of $(a_n)_n$;
\item $\lambda$ appears as an eigenvalue of $M(a)$ in a weighted automaton representation of $(a_n)_n$ of minimal dimension;
\item $1/ \lambda$ is a pole of $F$, that is, $q(1/\lambda)=0$.
\end{itemize}
Further, the characteristic roots appear as eigenvalues of $M(a)$ in \emph{every} representation of $(a_n)_n$ using a weighted automaton.
But in a weighted automaton that is not of minimal dimension, the matrix $M(a)$ may have additional eigenvalues. 

In characteristic $0$, every LRS $(a_n)_n$, has, for large enough $n$, a representation as an \emph{exponential polynomial sequence \textup(EPS\textup)}:
\[
a_n = \sum_{i=1}^k q_i(n) \lambda_i^n \quad\text{for sufficiently large n},
\]
with $q_i$ polynomials over $\overline{K}$ and $\lambda_i \in \overline{K}$ the nonzero characteristic roots of $(a_n)_n$.
Furthermore, the exponential bases $\lambda_i$ and the polynomials $q_i$ are uniquely determined by $(a_n)_n$.

\begin{example}
    The Fibonacci numbers admit the representation $F_n = \tfrac{1}{\sqrt{5}} \varphi^n - \tfrac{1}{\sqrt{5}} \psi^n$.
\end{example}

In characteristic $p > 0$, the situation is more complicated (see Appendix~\ref{subsec:appendix-exppoly}): an LRS may not have a representation by an exponential polynomial (even for large $n$).
If it does have such a representation, it is however still unique as long as the polynomials $q_i$ are chosen of minimal degree, that is, with $\deg(q_i) < p$.
Further, every EPS is an LRS, and the exponential bases of the EPS are precisely the nonzero characteristic roots of the LRS.

\subsection{Cost Register Automata}
We will introduce one more formalism that generalises weighted automata to polynomial updates~\cite{AlurDDRY13}.
A \emph{cost register automaton} (CRA) over a field $K$ is a tuple $\C = (Q, q_0, d, \Sigma, \delta, \mu, \nu)$, where: $Q$ is a finite set of states; $q_0 \in Q$ is the initial state; $d \in \N$ is its dimension; $\Sigma$ is a finite alphabet; $\delta\colon Q \times \Sigma \to Q \times \Poly^d$ is a deterministic transition function, where $\Poly^d$ is the set of $d$-dimensional polynomial maps; $\mu\colon K^d$ is the vector of initial register values; and $\nu\colon Q \to \Poly^d$ is the final function.
Here, a \emph{polynomial map} $P \in \Poly^d$ is a tuple $P=(p^1,\dots,p^d)$ with polynomials $p^i \in K[x_1,\dots,x_d]$.
Every polynomial map induces a function $K^d \to K^d$.

Given $q \in Q$ and $a \in \Sigma$ we write $p_{q,a}$ for the polynomial map such that $\delta(q,a) = (q',p_{q,a})$ for some $q' \in Q$.
Note that if we ignore the polynomials in $\delta$, then $(Q,q_0,\Sigma,\delta)$ is just a deterministic finite automaton without final states. Thus, given a word $w$, there is a unique state reachable from $q_0$ when reading $w$. We will denote it $q_w$.
For words $w \in \Sigma^+$ we define polynomial maps $p_w$ by induction: if $w = a \in \Sigma$ is a letter then $p_w = p_{q_0,a}$; otherwise if $w = w'a$ for a letter $a \in \Sigma$ then $p_w = p_{q_{w'},a} \circ p_{w'}$.

A CRA defines a function $\C\colon \Sigma^* \to K$, similarly to weighted automata. Formally, given a word $w = w_1 \ldots w_n \in \Sigma^*$ we define $\C(w) =  (\nu(q_w) \circ p_w)(\mu)$. 
\begin{example}
We can define the automaton recognising the same function as in \cref{example:linear}. There is only one state, which is also initial, and only one letter. The dimension is $2$, we will label the two resulting registers as $x$ and $y$. There is only one transition defined by the polynomial map $(p^x,p^y)$ with $p^x(x,y) = x+y$ and $p^y(x,y) = y$. The initial vector is defined by $\mu(x) = 0$, $\mu(y) = 1$; and the output is the polynomial $x$.
\end{example}

One can think of the polynomial maps as generalising linear updates definable by matrices. When restricting the model to linear polynomials, the CRA formalism is equivalent to weighted automata~\cite{AlurDDRY13}, and it is called linear CRA. The resulting weighted automaton is of polynomial size in the size of the linear CRA. Note that CRA use separate notions of states (the set $Q$) and registers (i.e., the dimension $d$). In general, states are not needed, as one can easily encode the states by enlarging the dimension to $d \times |Q|$, even for linear CRA. However, such encodings do not preserve the copyless restriction on CRA, which we discuss next.

A \emph{copyless} CRA (CCRA) is a CRA such that all polynomial maps in the transition function and the output function are copyless. A polynomial map $P \in \Poly^d$ is copyless if it can be written using sum, product, variable names and constants using each variable name only once. In particular $x^k$ for $k > 1$ is not copyless. 
\begin{example}
For $d = 3$ the map $P$ is defined by three polynomials $p^x(x,y,z)$, $p^y(x,y,z)$ and $p^z(x,y,z)$.
If $p^x = (x+3)\cdot (y+z)$, $p^y = 7$ and $p^z=1$, then $P$ is copyless; but if $p^x = y+1$, $p^y = y$ and $p^z=z$, then $P$ is not copyless.
\end{example}
It is easy to see that copyless polynomial maps are preserved under composition.
Thus, in a CCRA all polynomial maps $p_w$ are copyless.

Functions definable by Copyless CRA are known to be definable by weighted automata~\cite{MazowieckiR15,MazowieckiR19}. For intuition we give a very short self-contained proof of this fact in Appendix~\ref{ssec:ccra-to-wa}.

\section{Pumping Sequence Families of CCRA}
\label{sec:ccra}
Throughout the section, fix a field $K$.
In this section we prove a result restricting Pumping Sequence Families of CCRA.
This result will be based on the observation that sequences of the form $(\A(u w^{m(n+1)} v))_n$, obtained from a CCRA $\A$, are always representable by very particular exponential polynomials. 
To this end, we first introduce the following class of functions.

\begin{definition} \label{d:generable}
    A $K$-valued sequence $(a_n)_n$ is an \emph{exponential polynomial sequence generable from $A \subseteq K$} \textup(in short, an \emph{$A$-generable EPS}\textup) if it can be obtained, using pointwise products and sums, from the following sequence families:
    \begin{itemize}
      \item Constant sequences $(\alpha)_n$ for $\alpha \in A$,
      \item The linear sequence $(n \cdot 1_K)_n$,
      \item Exponential sequences \textup(that is, geometric progressions\textup) $(\alpha^n)_n$ for $\alpha \in A$,
      \item Sequences of the form $\big(\frac{1}{\alpha-1} \alpha^n - \frac{1}{\alpha - 1}\big)_n$ for $1 \neq \alpha \in A$.
    \end{itemize}
\end{definition}

The last family may be a bit unexpected at first glance.
It arises from the geometric sum
\[
\frac{1}{\alpha - 1} \alpha^n - \frac{1}{\alpha - 1} = \frac{\alpha^n - 1}{\alpha - 1} = \sum_{i=0}^{n-1} \alpha^i \qquad (\alpha \ne 1),
\]
with the representation in \cref{d:generable} corresponding to the normal form for exponential polynomials (with the two exponential bases $\alpha$ and $1_K$).
Because possibly $1/(1-\alpha) \not \in A$, this last family cannot always be generated from the other three families.

\begin{example}
    Since $\sum_{i=0}^{n-1} \alpha^i = \alpha\big( \cdots (\alpha (\alpha+1) + 1)\big) + 1$, the geometric sum appears when iterating a copyless update rule of the form $x \mapsto \alpha x + 1$ from the starting value $1$.
\end{example}

Taking $A=K$, the class of $K$-generable EPS admits a more familiar description.

\begin{lemma} \label{l:eps-vs-geneps}
    A sequence $(a_n)_n$ is a $K$-generable EPS if and only if it is an EPS with coefficients and exponential bases in $K$.
\end{lemma}

\begin{proof}
    We first check that every $K$-generable EPS is indeed an EPS.
    Since EPS with coefficients and exponential bases in $K$ are closed under products and sums, it suffices to verify the claimed property for the families in \cref{d:generable}.
    However, each of these families is obviously an EPS and the only exponential bases that appear are $1_K$ and $\alpha \in K$.

    Conversely, suppose that $(a_n)_n$ has a representation $a_n = \sum_{\lambda \in K^\times} \sum_{i \ge 0} \alpha_{\lambda,i} n^i \lambda^n$ with $\alpha_{\lambda,i} \in K$ (only finitely many of which are nonzero).
    Each of $(\alpha_{\lambda,i})_n$, $(n^i)_n$ and $(\lambda^n)_n$ is clearly a $K$-generable EPS, and so is therefore $(a_n)_n$.
\end{proof}

Recall that the exponential bases being contained in $K$ is a nontrivial restriction on an EPS.
In general, these will be contained in the algebraic closure $\overline K$.
In particular, every $A$-generable EPS is trivially a $K$-generable EPS, and hence by \cref{l:eps-vs-geneps} an EPS (in the sense discussed in \cref{sec:preliminaries}), so that our terminology is consistent.
Working with, possibly proper, subsets $A \subseteq K$ will be crucial to obtain a pumping-like criterion that is strong enough to differentiate between CCRA and polynomially ambiguous WFA.

We need a final definition before stating our main theorem of the section.

\begin{definition}
\label{definition:RCCRA}
Let $R \subseteq K$ be a subsemiring.
An \emph{$R$-CCRA} is a CCRA with all of its initial register values, output expression and transition coefficients in $R$.
\end{definition}

We will now exploit Pumping Sequence Families, the main tool introduced in this paper (recall \cref{d:psf}).

\begin{theorem}
\label{theorem:CCRA}
If $R \subseteq K$ is a subsemiring and $f\colon \Sigma^* \to K$ is recognised by an $R$-CCRA, then there exists $m \ge 1$ such that, 
\begin{itemize}
\item for every $g \in \PSF(f)$, the sequence $(h(n))_n = (g(m(n+1)))_n$ is an $R$-generable EPS, and
\item if the characteristic of $K$ is $0$ and $q$ is the exponential polynomial representing $h$, then for every $k \in \N$ the sum of $k$-degree coefficients $S_k(q)$ is in $R$.
\end{itemize}
\end{theorem}

The sum of $k$-degree coefficients $S_k(q)$ is obtained by summing all the coefficients of $x^k$ in $q$ across all the exponential bases (see Appendix~\ref{subsec:appendix-exppoly} for a detailed discussion).
The characteristic condition in the second property can be removed, leading to a slightly weaker result which is discussed in Appendix~\ref{ssec:ccra-positive-char}.

It is obvious that, for any input, the output of an $R$-CCRA is in $R$. However, this is different from the property in \cref{theorem:CCRA} --- we make a claim about the coefficients of the exponential polynomial, not the values that it takes. 
The individual coefficients do not need to always lie in $R$, as the following example illustrates.

Observe that we can assume $R$ is finitely generated -- by the initial register values, output expression and transition coefficients of the CCRA.

\begin{example} \label{example:basic_ccra}
\begin{itemize}
\item Consider the left $\Z$-CCRA in \autoref{fig:basic_ccra}. On words of the form $a^{n+1}$, this CCRA outputs $q(n)=\A(n+1)=\tfrac{3^{n+1}-1}{2} = \tfrac{3}{2} \cdot 3^n - \tfrac{1}{2} \cdot 1^n$.
Even though the automaton itself only uses integer coefficients, a denominator $2$ appears in the coefficients of $q$.
However, the sum of the coefficients is $S_0(q)=\tfrac{3}{2} - \tfrac{1}{2} = 1$, an integer.

\item The second $\Z$-CCRA in \autoref{fig:basic_ccra} outputs
\[
\B(a^{n+1})=\frac{5^{n+2} - 1}{4}(n+7) + 3^{n+1} = \big(\tfrac{25}{4}n + \tfrac{175}{4}\big) \cdot 5^{n} + 3 \cdot 3^n + \big(-\tfrac{1}{4}n - \tfrac{7}{4} \big) \cdot 1^n \eqqcolon q(n).
\]
Here $S_1(q) = \tfrac{25}{4} - \tfrac{1}{4} = 6 \in \Z$ and $S_0(q) = \tfrac{175}{4} + 3 - \tfrac{7}{4} = 45 \in \Z$.
\end{itemize}
\end{example}
\begin{figure}
\centering
\begin{minipage}[t]{.6\textwidth}
    \captionsetup{width=0.9\linewidth}
    \centering

    \begin{tikzpicture}
    \node[state] (p0) {};
    \node[left = 1.25cm of p0] (pin) {};
    \node[right = 0.75cm of p0] (pout) {};
        
    \path
    (pin) edge[->] node[above] {$x\coloneqq 0$} (p0)
    (p0) edge[->] [loop above] node[above] {$x\coloneqq 3x+1$} (p0)
    (p0) edge[->] node[above] {$x$} (pout);

    \node[state, right = 2cm of pout] (q0) {};
    \node[left = 1.5cm of q0] (qin) {};
    \node[right = 1.25cm of q0] (qout) {};
        
    \path
    (qin) edge[->] node[above] {$x, z\coloneqq 1$} node[below] {$y\coloneqq 6$} (q0)
    (q0) edge[->] [loop above] node[above, align=center] {$x\coloneqq 5x+1$, $y \coloneqq y+1$,\\ $z \coloneqq 3z$} (q0)
    (q0) edge[->] node[above] {$xy+z$} (qout);
    \end{tikzpicture}
            
    \captionof{figure}{Two simple single-state CCRAs on a single-letter alphabet (\cref{example:basic_ccra}).}
    \label{fig:basic_ccra}
\end{minipage}%

\end{figure}

Before proving \cref{theorem:CCRA}, we demonstrate how it can be applied.
We use it to show that not every function recognisable by a polynomially ambiguous WFA can be recognised by a CCRA.

\begin{figure}
    \begin{minipage}[t]{.42\textwidth}
        \captionsetup{width=0.9\linewidth}
       
        \centering
        \begin{tikzpicture} 
        \node[state] (p0) {};
        \node[left = 0.4cm of p0] (pin) {};
        \node[state, right = 0.9cm of p0] (p1) {};
        \node[state, right = 0.9cm of p1] (p2) {};
        \node[right = 0.4cm of p2] (pout) {};
            
        \path
        (pin) edge[->] node[above] {$1$} (p0)
        (p0) edge[->] [loop above] node[above,align=center] {$a \mid 2$\\$b \mid 1$} (p0)
        (p0) edge[->] node[above] {$a \mid 2$} (p1)
        (p1) edge[->] [loop above] node[above,align=center] {$a \mid 2$\\$b \mid 1$} (p1)
        (p1) edge[->] node[above] {$b \mid 1$} (p2)
        (p2) edge[->] [loop above] node[above,align=center] {$a \mid 1$\\$b \mid 1$} (p2)
        (p2) edge[->] node[above] {$1$} (pout);
        \end{tikzpicture}  
            
        \captionof{figure}{A polynomially ambiguous weighted automaton with no equivalent CCRA (\cref{exm:wfa-not-ccra}).}
        \label{fig:polyamb}    
    \end{minipage}
    \begin{minipage}[t]{.57\textwidth}
        \captionsetup{width=0.9\linewidth}
        \centering
    
        \begin{tikzpicture}
            \node[state] (p0) {};
            \node[left = 1.75cm of p0] (pin) {};
            \node[right = 1.25cm of p0] (pout) {};
                
            \path
            (pin) edge[->] node[above] {$x,y,z\coloneqq 1$} (p0)
            (p0) edge[->] [loop above] node[above,align=center] {$x \coloneqq 2x + y$, $y \coloneqq 4$,\\ $z \coloneqq \frac{z}{2}+1$} (p0)
            (p0) edge[->] node[above] {$x+z$} (pout);
        
            \node[rectangle, draw=red, thick, minimum size=1.5em,
                  right = 1.75cm of p0, yshift = 1.5cm] (ry) {$y$};
            \node[rectangle, draw=blue, thick, minimum size=1.5em,
                  right = 0.75cm of ry] (rx) {$x$};
            \node[rectangle, draw=blue, thick, minimum size=1.5em,
                  below = 0.75cm of rx] (rz) {$z$};
        
            \path
            (ry) edge[->] (rx)
            (rx) edge[->] [loop above] (rx)
            (rz) edge[->] [loop above] (rz);
        \end{tikzpicture}    
       
        \captionof{figure}{A simple CCRA and its variable flow graph. Red nodes are constant registers; blue nodes are updating ones (\autoref{definition:simple}).}
        \label{fig:simple}    
    \end{minipage}
\end{figure}
    
\begin{example} \label{exm:wfa-not-ccra}    
    The automaton in \autoref{fig:polyamb} is polynomially ambiguous.
    Let $f\colon \{a,b\}^* \to \Q$ be the function associated with the automaton and, for the sake of contradiction, assume that $f$ can be recognised by a CCRA. 
    This yields a finitely generated subsemiring $R$ and a natural number $m$ such that for all $g \in \PSF(f)$, the sequence $(h(n))_n \coloneqq (g((n+1)m))_n$ meets the conditions from \autoref{theorem:CCRA}. Consider, for any $k$, $g_k \coloneqq \hat{f}(\varepsilon, a^kb, \varepsilon) \in \PSF(f)$. 
    By grouping the paths based on which $b$ is used to transition between the second and third state, we get
    \[
    h_k(n) \coloneqq g_k((n+1)m) = k2^k + 2k2^{2k} + \dots + {m(n+1)}k2^{{m(n+1)}k} = \sum_{j=1}^{m(n+1)} jk 2^{jk}.
    \]
    Using the identity $\sum_{j=1}^l j x^j = \tfrac{l x^{l+2} - (l+1) x^{l+1} + x}{(x-1)^2}$, which can be derived from the geometric sum $\sum_{j=1}^l x^j = \tfrac{x^{l+1}-x}{x-1}$ by formal differentiation and some easy manipulations, one finds
    \[
    h_k(n) =  q_1(n) \cdot (2^{km})^n + q_2(n) \cdot 1^n,
    \]
    with 
    \[
    q_1(n) =
    \frac{ km 2^{km + k}}{(2^k-1)} n 
    + \frac{k(m 2^{k}  - m - 1)2^{km+k}}{(2^k-1)^2} \quad\text{and}\quad q_2(n)=\frac{k2^k}{(2^k-1)^2}.
    \]

    We have $S_1(h_k) = km 2^{km+k}/(2^k-1)$.
    Only a finite set of prime numbers can appear among denominators of elements of $R$ and only finitely many primes divide $2m$, we can thus take a prime $p$ that fulfills neither of these conditions. 
    We can also now fix $k = p-1$.
    We get
    \[
    S_1(h_{p-1})=\frac{(p-1)m2^{(p-1)m + p - 1}}{2^{p-1} - 1}.
    \]
    Since $p \nmid 2m$, the numerator is not divisible by $p$. 
    However, by Fermat's Little Theorem, the denominator is.
    As we have assumed that $p$ does not appear in the denominator of any element of $R$, this means $S_1(h_{p-1}) \not\in R$, contradicting the statement of \autoref{theorem:CCRA}.
\end{example}

We record the conclusion as a theorem.

\begin{theorem}
    If $\card{\Sigma} \ge 2$, then there exist functions $f \colon \Sigma^* \to \Q$ that are recognisable by a polynomially ambiguous weighted automaton, but not by a $\Q$-CCRA.
\end{theorem}

\begin{proof}
    By \cref{exm:wfa-not-ccra}.
\end{proof}

\subsection{The Proof of \autoref{theorem:CCRA}}

The proof of \cref{theorem:CCRA} proceeds in several steps.
We start with a very simple case and then, in each step, use the previous result to show a slightly more general one.

\begin{definition}
\label{definition:simple}
A single letter, single state CCRA is \emph{simple} if, in the transition, every register value is either set to a constant \textup(\emph{constant registers}\textup), or depends only on its old value and on the values of constant registers \textup(\emph{updating registers}\textup).
\end{definition}

Since there is only one state, there is also only one transition, so the definition makes sense.
We can visualise constant and updating registers with a graph representing the dependency of register values on each other (\autoref{fig:simple}).

\begin{lemma}
\label{lemma:ccra_simple}
If $\A$ is a single state, single letter simple $R$-CCRA, then $(\A( a^{n+1}))_n$ is an $R$-generable EPS.
\end{lemma}
\begin{proof}
In a simple CCRA, the register values change in very simple ways.
For constant registers, after the first transition, they remain set to the same values.
For updating registers, the first update is special. However, after that, the input they get from the constant registers stabilizes. 
Let us consider what happens from that point on. After the first step, the register values are of course still in $R$. Since the updating registers can only depend on themselves and constants, the update formulas can be reduced to the form $x \coloneqq \alpha x + \beta$ for constants $\alpha$,~$\beta \in R$. This gives us an LRS ($x_{n+1} = \alpha x_n + \beta$).
Solving the LRS, we need to distinguish two cases.
For $\alpha = 1$, the solution is 
\[
x_{n+1} = x_1 + \beta n, \qquad\text{and for $\alpha \neq 1$ it is}\qquad
x_{n+1} = \alpha^{n} \frac{\beta}{\alpha - 1} - \frac{\beta}{\alpha - 1} + \alpha^{n} x_1.
\]
In both cases the sequence $(x_{n+1})_n$ is clearly an $R$-generable EPS.
Constant registers, leading to constant sequences, also clearly are $R$-generable EPS.

The output expression combines these sequences using sums and products of the sequences and additional constants from $R$. 
These operations preserve the property of being an $R$-generable EPS, and so the sequence $(\A(a^{n+1}))_n$ is an $R$-generable EPS.
\end{proof}

In the next two lemmas we will reason about the behaviour of CCRA on cycles. Similar, but different, observations were made in~\cite[Proposition~1 and Lemma~4]{MazowieckiR19}.

\begin{lemma}
\label{lemma:multi_reg}
If $\A$ is a single state, single letter $R$-CCRA \textup(not necessarily simple\textup) with $r$ registers, then $(\A( a^{r!(n+1)}))_n$ is an $R$-generable EPS.
\end{lemma}
\begin{proof}
Consider the compound effect on registers of the letter $a$ being applied $r!$ times. We can get the corresponding expressions simply by composing the substitution $r!$ times. They will still of course be copyless and polynomial, meaning we can create an auxiliary CCRA $\B$ with a 1-letter alphabet such that $\B( a^n) = \A( a^{nr!} )$. It is also easy to see, from how substitutions compose, that $\B$ is still an $R$-CCRA.

\begin{figure}
\centering
\begin{tikzpicture}
    \node[rectangle, draw, minimum size=1.5em] (Aq) {$q$};
    \node[rectangle, draw, minimum size=1.5em,right = 0.75cm of Aq] (Ar) {$r$};
    \node[rectangle, draw, minimum size=1.5em,right = 0.75cm of Ar] (Ay) {$y$};
    \node[rectangle, draw, minimum size=1.5em,right = 0.75cm of Ay] (Ax) {$x$};
    \node[rectangle, draw, minimum size=1.5em,right = 0.375cm of Ay.center, yshift = -1cm] (Az) {$z$};

    \path
    (Aq) edge[->] (Ar)
    (Ar) edge[->] (Ay)
    (Ay) edge[->] (Ax)
    (Ax) edge[->] (Az)
    (Az) edge[->] (Ay);

    \node[rectangle, draw=blue, thick, minimum size=1.5em, right = 2cm of Ax] (By) {$y$};
    \node[rectangle, draw=blue, thick, minimum size=1.5em, right = 0.75cm of By] (Bx) {$x$};
    \node[rectangle, draw=blue, thick, minimum size=1.5em, right = 0.75cm of Bx] (Bz) {$z$};

    \node[rectangle, draw=red, thick, minimum size=1.5em, below = 0.75cm of Bx] (Bq) {$q$};
    \node[rectangle, draw=red, thick, minimum size=1.5em, below = 0.75cm of Bz] (Br) {$r$};

    \path
    (Bq) edge[->] (Bx)
    (Br) edge[->] (Bz)
    (By) edge[->] [loop above] (By)
    (Bx) edge[->] [loop above] (Bx)
    (Bz) edge[->] [loop above] (Bz);
\end{tikzpicture}    
\caption{An example for variable flow graphs of $\A$ and $\B$ in the proof of \autoref{lemma:multi_reg}.}
\label{fig:dependency}
\end{figure}

We claim that the new CCRA $\B$ is simple: we will prove this by looking at the variable flow graph of $\A$.
Since $\A$ is copyless, there is at most one outgoing edge from any vertex.
In $\B$ the expression for register $v$ will use $u$ if and only if in the variable flow graph of $\A$ there is a path of length $r!$ from $u$ to $v$. (This is visualised in \autoref{fig:dependency}.)

Consider an arbitrary register $t$. It will either be in a cycle on the graph of $\A$ or not. Assume $t$ is not in a cycle and there is a path of length $r!$ from some register $u$ to $t$. Such a path would have to contain a cycle. However, that is impossible, since each vertex has at most one outgoing edge and $t$ itself is not in a cycle. This means $t$ will be a constant register in the auxiliary automaton. 

Now assume $t$ is in a cycle in the variable flow graph of $\A$, and let $l$ be the length of the cycle.
We want to prove that $t$ is an updating register in $\B$.
Assume there is a path of length $r!$ from some $u$ to $t$.
To show that $t$ is an updating register, we need to show that either $u=t$ or $u$ is a constant register in $\B$.
If $u$ is in the same cycle as $t$, we have $u = t$, since $l \mid r!$.
If $u$ is outside the cycle containing $t$, then $u$ cannot be a part of any cycle, as any vertex can have at most one outgoing edge.
This means that $u$ is a constant register in $\B$.
Thus, the auxiliary CCRA $\B$ is simple, and we can apply \autoref{lemma:ccra_simple} to it, finishing the proof.
\end{proof}

\begin{lemma}
\label{lemma:multi_state}
If $\A$ is a single letter $R$-CCRA \textup(not necessarily single state\textup) with $s$ states and $r$ registers, then $(\A(a^{(4r+2)!s!(n+1)}))_n$ is an $R$-generable EPS.
\end{lemma}
\begin{proof}
Consider the compound effect on registers of the letter $a$ being applied $s!$ times. We can get the corresponding transitions between states by looking at paths of length $s!$, and corresponding update expressions by composing appropriate $s!$ substitutions. The updates will of course still be copyless and polynomial, and the transitions deterministic, meaning we can create an auxiliary CCRA $\B$ such that $\B( a^{n}) = \A( a^{ns!} )$. Note that the transition expression coefficients will all still be in $R$, so $\B$ is still an $R$-CCRA. Since $\A$ is deterministic, after at most $s$ steps it always reaches a cycle.
This cycle has length at most $s$, and so its length divides $s!$.
This means that, after trimming $\B$, we get an automaton of one of the forms in \autoref{fig:ccra_trimmed}.

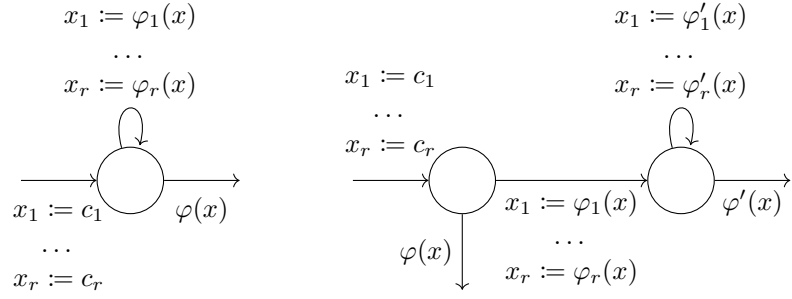
\begin{figure}
\centering
\begin{tikzpicture}
\node[state] (p0) {};
\node[left = 1cm of p0] (pin) {};
\node[right = 1cm of p0] (pout) {};

\node[right = 1cm of pout] (qin) {};
\node[state, right = 1cm of qin] (q0) {};
\node[state, right = 2cm of q0] (q1) {};
\node[right = 1cm of q1] (qout) {};
\node[below = 1cm of q0] (q0ut) {};
   
\path
(pin) edge[->] node[below = 0.2cm, align=center] {$x_1 \coloneqq c_1$\\$\dots$\\$x_r \coloneqq c_r$} (p0)
(p0) edge[->] [loop above] node[above, align=center] {$x_1 \coloneqq \varphi_1(x)$\\$\dots$\\$x_r \coloneqq \varphi_r(x)$} (p0)
(p0) edge[->] node[below = 0.1cm] {$\varphi(x)$} (pout)
(qin) edge[->] node[above = 0.2cm, align=center] {$x_1 \coloneqq c_1$\\$\dots$\\$x_r \coloneqq c_r$} (q0)
(q0) edge[->] node[below, align=center] {$x_1 \coloneqq \varphi_1(x)$\\$\dots$\\$x_r \coloneqq \varphi_r(x)$} (q1)
(q1) edge[->] [loop above] node[above, align=center] {$x_1 \coloneqq \varphi'_1(x)$\\$\dots$\\$x_r \coloneqq \varphi'_r(x)$} (q1)
(q1) edge[->] node[below] {$\varphi'(x)$} (qout)
(q0) edge[->] node[left] {$\varphi(x)$} (q0ut)
;
\end{tikzpicture}
\caption{The two possible forms the auxiliary automaton $\B$ can take in the proof of \autoref{lemma:multi_state}.}\label{fig:ccra_trimmed}
\end{figure}

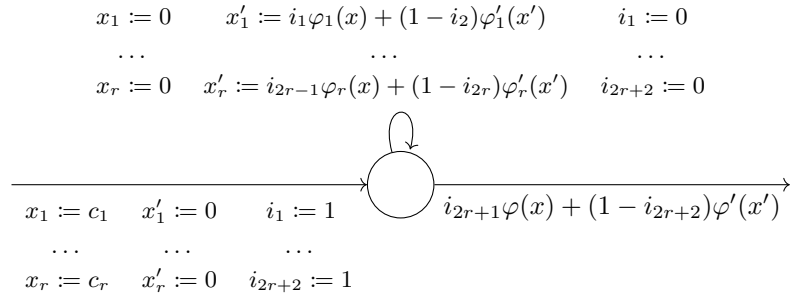
\begin{figure}
\centering
\begin{tikzpicture}
\node[state] (p0) {};
\node[left = 4.7cm of p0] (pin) {};
\node[right = 4.7cm of p0] (pout) {};
   
\path
(pin) edge[->] node[below, align=center] {
\begin{tabular}{ c c c }
 $x_1 \coloneqq c_1$ & $x_1' \coloneqq 0$ & $i_1 \coloneqq 1$ \\ 
 $\dots$ & $\dots$ & $\dots$ \\  
 $x_r \coloneqq c_r$ & $x_r' \coloneqq 0$ & $i_{2r+2} \coloneqq 1$    
\end{tabular}} (p0)
(p0) edge[->] [loop above] node[above, align=center] {
\begin{tabular}{ c c c}
 $x_1 \coloneqq 0$ & $x_1' \coloneqq i_1 \varphi_1(x) + (1-i_2) \varphi'_1(x')$ & $i_1 \coloneqq 0$ \\ 
 $\dots$ & $\dots$ & $\dots$ \\  
 $x_r \coloneqq 0$ & $x_r' \coloneqq i_{2r-1}\varphi_r(x) + (1-i_{2r}) \varphi'_r(x')$ & $i_{2r+2} \coloneqq 0$ 
\end{tabular}} (p0)
(p0) edge[->] node[below] {$i_{2r+1} \varphi(x) + (1-i_{2r+2}) \varphi'(x')$} (pout)
;
\end{tikzpicture}
\caption{How to transform $\B$ into a single state automaton in the proof of \autoref{lemma:multi_state}.}\label{fig:ccra_single_state}
\end{figure}

We want to reduce $\B$ to only one state.
The first possible form already has only one state. 
The second one can easily be simulated with one state, as shown in \autoref{fig:ccra_single_state}.
After this operation, the automaton $\B$ is a single-state $R$-CCRA with $4r+2$ registers such that $\B( a^n ) = \A(a^{ns!} )$. This lets us use \autoref{lemma:multi_reg} and finishes the proof.
\end{proof}

\begin{lemma}
\label{lemma:ccra_many_letters}
If $\A$ is an $R$-CCRA \textup(not necessarily single letter\textup) with $r$ registers and $s$ states, then, for all $w \in \Sigma^*$, the sequence $(\A(w^{(4r+2)!s!(n+1)}))_n$ is an $R$-generable EPS.
\end{lemma}
\begin{proof}
Consider the composite effect of the word $w$ on registers and state transitions. This effect is still copyless, polynomial, deterministic and all the transition coefficients are still in $R$.
We can thus create an auxiliary $R$-CCRA $\B$ such that $\A ( w^n ) = \B ( a^n )$.
The CCRA $\B$ has a one letter alphabet, letting us use \autoref{lemma:multi_state} and finishing the proof.
\end{proof}

\begin{lemma}
\label{lemma:general_ccra_generable}
If $\A$ is an $R$-CCRA with $r$ registers and $s$ states, then, for all $u, w, v \in \Sigma^*$, the sequence $(\A( uw^{(4r+2)!s!(n+1)}v ))_n$ is an $R$-generable EPS.
\end{lemma}
\begin{proof}
Adding some prefix $u$ simply changes the initial register values.
The register values are of course still in $R$.
Adding a suffix $v$ simply changes the output expression. 
Its coefficients are of course still in $R$.
Thus, we obtain an $R$-CCRA $\B$ with $\B(x) = \A(uxv)$ for all words $x \in \Sigma^*$.
By \autoref{lemma:ccra_many_letters}, the sequence $(\A(uw^nv))_n = (\B(w^n))_n$ is an $R$-generable EPS.
\end{proof}

We also need the next lemma which is proven in Appendix~\ref{sec:additional-proofs}.

\begin{restatable}[]{lemma}{coeffsum}
\label{lemma:coeff_sum}
If $R \subseteq K$ is a subsemiring, $\chr K = 0$, $(a_n)_n$ is an $R$-generable EPS and $q$ is the exponential polynomial representing $(a_n)_n$, the sum of $k$-degree coefficients of $q$ is in $R$.
\end{restatable}
We can finally prove the main theorem of \cref{sec:ccra}.

\begin{proof}[Proof of \autoref{theorem:CCRA}]
Let $\A$ be an $R$-CCRA recognizing $f$.
Let $m = (4r+2)!s!$, where $r$ is the number of registers and $s$ is the number of states of $\A$. 
By \autoref{lemma:general_ccra_generable} the sequence $(h(n))_n \coloneqq (g((n+1)m))_n \coloneqq (f(uw^{m(n+1)}v))_n$ is an $R$-generable EPS.
Let $q$ be the exponential polynomial representing $h$.
By \autoref{lemma:coeff_sum}, for every $k$, the sum of $k$-degree coefficients of this representation is in $R$.
\end{proof}

\subsection{CCRA versus \texorpdfstring{$R$}{R}-generable EPS}\label{subsection:examples}

We have seen that, for functions $f\colon \Sigma^* \to K$ recognisable by a CCRA, there always exists a finitely generated subsemiring $R$ and $m \ge 1$ such that for all $g \in \PSF(f)$ the sequence $(g((n+1)m))_n$ is an $R$-generable EPS.
We conjecture that this is not sufficient to characterise functions recognised by CCRA, even if it is already known that the function is recognised by a weighted automaton.

At present, we do not have a counterexample, but we outline a plausible candidate in this subsection.
However, it appears difficult to prove that the given function is not recognised by a CCRA.

\begin{example}
\label{example:maybe_unrecognisable}
Consider the following function $f\colon \{0,1\}^* \to \Q$. Given $w \in \Sigma^*$, let $0 < k_1 < k_2 <\ldots < k_r$ be the indices of all $1$'s in $w$, e.g.\ for $w = 0110$ we have: $r=2$, $k_1 = 2$, $k_2 = 3$. Then $f(w) = \sum_{i=1}^r k_i$.

Technical computations show that each element of $\PSF(f)$ is a $\tfrac{1}{2}\Z$-generable EPS (see Appendix~\ref{sec:additional-proofs}).
Nevertheless, it appears unlikely to us that $f$ could be recognised by a CCRA.
\end{example}

The reason why functions can or cannot be recognised by a CCRA can be subtle: while it appears that the function $f$ in \cref{example:maybe_unrecognisable} cannot be recognised by a CCRA, the following example shows a function of similar nature that can be recognised by a CCRA.

\begin{example} \label{exm:looks-unrecognisable-but-is-recognisable}
We define $g\colon \{0,1\}^* \to \Q$: as before, given $w \in \Sigma^*$, let $0 < k_1 < k_2 <\ldots < k_r$ be the indices of all $1$'s. Then $g(w) = \sum_{i=1}^{r} 2^{k_i}$
is recognised by the CCRA in \cref{fig:ccra_binary}.
\end{example}

\begin{figure}
\centering
\begin{minipage}[t]{0.5\textwidth}
    \captionsetup{width=0.9\linewidth}
    \centering
    \begin{tikzpicture}
    \node[state] (p0) {$p$};
    \node[left = 1.5cm of p0] (pin) {};
    \node[right = 1cm of p0] (pout) {};
        
    \path
    (pin) edge[->] node[above] {$x \coloneqq 1$} node[below] {$y \coloneqq 0$} (p0)
    (p0) edge[->] [loop above] node[above, align = center] {$0$:\ \ $x \coloneqq 2x$, $y \coloneqq \frac{y}{2}$} (p0)
    (p0) edge[->] [loop below] node[below, align = center] {$1$:\ \ $x \coloneqq 2x$, $y \coloneqq \frac{y}{2} + 1$} (p0)
    (p0) edge[->] node[above] {$\frac{yx}{2}$} (pout);
    \end{tikzpicture}
    \captionof{figure}{A CCRA recognising a function that, at first glance, may seem unrecognisable by a CCRA (\cref{exm:looks-unrecognisable-but-is-recognisable}).}\label{fig:ccra_binary}
\end{minipage}%
\begin{minipage}[t]{0.5\textwidth}
    \captionsetup{width=0.9\linewidth}
    \centering
    \begin{tikzpicture}
    \node[state] (p0) {$p$};
    \node[left = 1.5cm of p0] (pin) {};
    \node[right = 1cm of p0] (pout) {};
        
    \path
    (pin) edge[->] node[above] {$x\coloneqq 1$} node[below]{$y \coloneqq 0$} (p0)
    (p0) edge[->] [loop above] node[above, align=center] {$a$:\ \ $x \coloneqq x$, $y \coloneqq y+1$} (p0)
    (p0) edge[->] [loop below] node[below, align=center] {$b$:\ \ $x \coloneqq xy$, $y \coloneqq 0$} (p0)
    (p0) edge[->] node[above] {$x$} (pout)
    ;
    \end{tikzpicture}
    \captionof{figure}{A two-register CRA recognising the function from \cref{exm:ccra-not-pa}.}\label{fig:ccra_product_of_ks}
\end{minipage}
\end{figure}

Another promising example is discussed in Appendix~\ref{sec:additional-proofs}.

\section{Pumping Sequence Families of Polynomially Ambiguous WA}\label{sec:polyamb}

In this section we prove a result restricting Pumping Sequence Families of polynomially ambiguous weighted automata.

\begin{theorem}
\label{t:poly_fin_gen}
If $f$ is recognised by a polynomially ambiguous weighted automaton over $K$, then there exist a finitely generated multiplicative semigroup $G \subseteq \overline{K}$ and $N \ge 1$ such that
\begin{itemize}
    \item the characteristic roots of every sequence in $\PSF(f)$ are contained in $G$,
    \item and $\alpha^{N} \in K$ for all $\alpha \in G$.
\end{itemize}
\end{theorem}

With this theorem, we can give the following example. 

\begin{example} \label{exm:ccra-not-pa}
The function $f\colon \{a,b\}^* \to \Q$ defined by the CCRA in \cref{fig:ccra_product_of_ks} cannot be recognised by a polynomially ambiguous weighted automaton.
\end{example}
\begin{proof}
Let us consider inputs of the form $(a^kb)^{n}$ for $k \ge 1$. 
We have $\hat{f}(\varepsilon, a^kb, \varepsilon)(n) = f\big( (a^kb)^{n} \big) = k^n$, meaning that every natural number $k$ appears as a characteristic root of an LRS in $\PSF(f)$.
However, the monoid $(\Z_{>0},\cdot)$ generates $(\Q_{>0},\cdot)$ as a group, and $(\Q_{>0},\cdot)$ is a countably generated free abelian group (with primes as the generators).
Since subgroups of a finitely generated abelian group are finitely generated, but there are infinitely many primes, the natural numbers cannot be a submonoid of a finitely generated abelian group.
\Cref{t:poly_fin_gen} shows that $f$ is not recognised by a polynomially ambiguous weighted automaton.
\end{proof}

We again record this conclusion as a theorem.
\begin{theorem}
    If $\card{\Sigma} \ge 2$, then there exist functions $f\colon \Sigma^* \to \Q$ that can be recognised by a $\Q$-CCRA but not by a polynomially ambiguous weighted automaton over $\Q$.
\end{theorem}

\begin{proof}
    By \cref{exm:ccra-not-pa}.
\end{proof}

The core of the proof of \cref{t:poly_fin_gen} will be the following lemma.
The argument is similar to an argument in \cite{JeckerMP24} and in \cite[Prop.~9.3]{puch-smertnig24}. A self-contained proof is in Appendix~\ref{sec:additional-proofs}.
\begin{restatable}[]{lemma}{polyambuppertriangular}
\label{l:polyamb_upper_triangular}
Let $\A=(d,I,M,F)$ be a trim polynomially ambiguous weighted automaton. 
Then, for every word $w \in \Sigma^*$, there exists a permutation matrix $P$ such that $P \cdot M( w^{d!} ) \cdot P^{-1}$ is upper triangular. 
Furthermore, all nonzero eigenvalues of $M(w^{d!})$ are products of transition weights \textup(that is, of entries of the matrices $M(a)$ for letters $a \in \Sigma$\textup).
\end{restatable}

Before proving \cref{t:poly_fin_gen}, we need a final small observation.

\begin{lemma} \label{l:fg-roots}
    If $H \subseteq K$ is a finitely generated semigroup and $N \ge 1$, then the semigroup $G=\{\, \alpha \in \overline{K} \mid \alpha^N \in H \cup \{1\} \,\}$ is also finitely generated.
\end{lemma}

\begin{proof}
    Suppose $\beta_1$, \dots,~$\beta_n$ generate $H$.
    For each $\beta_i$ let $\alpha_i \in \overline{K}$ be a root of $X^N - \beta_i$.
    Let $G'$ be the subsemigroup of $\overline{K}$ generated by $\alpha_1$, \dots,~$\alpha_n$ together with the $N$-th roots of unity in $\overline{K}$ (of which there are at most $N$, since they are the roots of $X^N-1$). 
    
    We claim $G=G'$.
    The inclusion $G' \subseteq G$ holds by definition.
    Suppose $\gamma \in G$. 
    Then $\gamma^N = \beta_1^{k_1}\cdots \beta_n^{k_n}$ for some $k_i \ge 0$.
    Define $\gamma'\coloneqq \alpha_1^{k_1} \cdots \alpha_n^{k_n} \in G'$.
    Then $(\gamma')^N = \gamma^N$.
    It follows that $\gamma= \gamma' \zeta$ with $\zeta$ an $N$-th root of unity (whether or not $\gamma =0$).
    So $\gamma \in G'$.
\end{proof}

\begin{proof}[Proof of \cref{t:poly_fin_gen}]
We have to show that there exists a finitely generated multiplicative semigroup $G \subseteq \overline{K}$ and an $N \ge 1$ such that for every $u$, $w$,~$v \in \Sigma^*$, the characteristic roots of $(f(uw^nv))_n$ are contained in $G$ and $\alpha^{N} \in K$ for all $\alpha \in G$.

Let $\A=(d,I,M,F)$ be a polynomially ambiguous weighted automaton recognising $f$.
Let $N \coloneqq d!$.
Without restriction, we can take $\A$ to be trim.
Let $H \subseteq K$ be the subsemigroup of $K$ generated by all the finitely many transition weights of $\A$, and let $G = \{\, \alpha \in \overline{K} \mid \alpha^N \in H \cup \{1\} \,\}$.
By \cref{l:fg-roots}, the semigroup $G$ is finitely generated.

The characteristic roots of the LRS $(\A(uw^{n}v))_n$ are eigenvalues of $M(w)$.
By \cref{l:polyamb_upper_triangular}, the eigenvalues of $M(w^{N})$ are products of transition weights.
We have $M( w^{N} )= M(w)^{N}$, and so the eigenvalues of $M(w)$ are roots of degree $N$ of products of transition weights of $\A$. This means they belong to $G$.
\end{proof}

We (ambitiously) conjecture the following converse of \cref{t:poly_fin_gen}.
\begin{conjecture} \label{c:inverse-pumping}
    Let $f \colon \Sigma^* \to K$ be recognised by a weighted automaton.
    If there exists a finitely generated multiplicative subsemigroup $G \subseteq \overline{K}$ and $N \in \Z_{\ge 1}$ such that
    \begin{itemize}
    \item the characteristic roots of every sequence in $\PSF(f)$ are contained in $G$,
    \item and $\alpha^{N} \in K$ for all $\alpha \in G$,
    \end{itemize}
    then $f$ is recognised by a polynomially ambiguous weighted automaton.
\end{conjecture}

\cref{c:inverse-pumping} postulates a pumping-style characterisation.
The following conjecture postulates a ``global'' characterisation, with a similar restriction as in \cref{c:inverse-pumping} imposed on the eigenvalues of the matrix semigroup.
Here it is important that the condition is imposed on \emph{all} matrices, not just on the generators.

\begin{conjecture} \label{c:inverse-global}
    Let $f \colon \Sigma^* \to K$ be recognised by a weighted automaton.
    If there exists a finitely generated multiplicative subsemigroup $G \subseteq \overline{K}$ and $N \ge 1$ such that
    \begin{itemize}
        \item all eigenvalues of matrices $M(w)$ for $w \in \Sigma^*$ are contained in G,
        \item and $\alpha^N \in K$ for all $\alpha \in G$,
    \end{itemize}
    then $f$ is recognised by a polynomially ambiguous weighted automaton.
\end{conjecture}

A positive resolution of the conjectures would extend a characterisation in similar spirit of functions that can be recognised by unambiguous weighted automata \cite{BellS23}.
While the conjectures seem ambitious, in the preprint \cite{puch-smertnig24}, \cref{c:inverse-global} was already proved in the case that all transition matrices are invertible.

The following lemma shows that Conjectures~\ref{c:inverse-pumping} and~\ref{c:inverse-global} are in fact equivalent.

\begin{lemma}
Let $\A=(d,I,M,F)$ be a minimal weighted automaton and let $w \in \Sigma^*$.
Then the set of nonzero eigenvalues of $M(w)$ is precisely the set of all nonzero characteristic roots of the LRS $(\A(uw^nv))_n$ as $u$,~$v \in \Sigma^*$ range through all words. 
\end{lemma}
\begin{proof}
One direction is obvious --- characteristic roots come from eigenvalues of the matrix $M(w)$.
We only have to show that every nonzero eigenvalue $\lambda \in \overline{K}$ of $M(w)$ shows up as characteristic root of some LRS.

Working over $\overline{K}$ we can assume that $K=\overline{K}$ is algebraically closed.
This allows us to change to a basis in which $M(w)$ is in the Jordan normal form.
In particular, we can assume that $M(w)$ is upper triangular and $M(w)[1,1] = \lambda$.
Let $e_1=(1,0,\dots,0) \in K^{d \times 1}$ and let $e_1^\intercal$ be its transpose.
Then $\lambda^n = e_1^\intercal M(w^n) e_1$.

Because $\A$ is minimal, the reachability set $\{\, I^\intercal M(w) \mid w \in \Sigma^* \,\}$ spans $K^{1 \times d}$ as a vector space, and analogously the coreachability set spans $K^{d \times 1}$ -- otherwise we could easily decrease the dimension.
Therefore, there exist $\alpha_i$,~$\beta_j \in K$ and $u_i$,~$v_j \in \Sigma^*$ such that $e_1^\intercal = \sum_{i=1}^d \alpha_i I^\intercal M(u_i)$ and $e_1 = \sum_{j=1}^d M(v_j) F \beta_j$.
Now
\[
\lambda^n = \sum_{i=1}^d \sum_{j=1}^d \alpha_i \beta_j\,  I^\intercal M(u_i w^n v_j)F,
\]
expresses the LRS $(\lambda^n)_n$ as linear combination of LRS $(I^\intercal M(u_i w^n v_j)F)_n$.
Since the former has a characteristic root $\lambda$, a summand must have $\lambda$ as a characteristic root as well: this is easily seen by considering the LRS as rational functions, and recalling that nonzero characteristic roots correspond to reciprocals of poles, or by the uniqueness result in \cref{t:lrs-repr}.
\end{proof}

While the main theorem of this section provides a necessary pumping criterion for polynomially ambigualisable automata, that is, those weighted automata that are equivalent to polynomial ambiguous ones, another related open problem is to relate the minimal degree of the polynomial bounding the ambiguity to arithmetic properties of the output (in other words, to characterise linearly ambiguous, quadratically ambiguous, etc.).
At least in characteristic zero, a tempting idea is to look at the degrees of polynomials arising in the PSF. 
Indeed, it is easy to see that if the ambiguity of the automaton is bounded by a polynomial of degree $d$, then no polynomial of higher degree can appear in the PSF.
The converse however does not hold, as the following example shows.

\begin{example}
The function $f\colon \{0,1\}^* \to \Q$, mapping a binary word to the natural number it represents (say, LSB on the left), is easily seen to be recognisable by a weighted automaton. 
We check that, for any $u$, $w$,~$v \in \Sigma^*$, the exponential polynomial representation of $(\A(uw^nv))_n$ only contains constant polynomials. Indeed, let $u = u_1 u_2 \dots u_r$, $w = w_1 w_2 \dots w_t$, $v = v_1 v_2 \dots v_l$. We have
{\small
\[
\begin{split}
f(uw^nv) &= 2^0 u_1 + 2^1 u_2 + \dots + 2^{r-1}u_r + 
            ( 2^r w_1 + 2^{r+1} w_2 + \dots + 2^{r+t-1} w_t) ( 1 + 2^t + \dots + 2^{t(n-1)} ) \\
         &+ 2^{nt+r} v_1 + \dots + 2^{nt+r+l-1}v_l = \alpha + \beta \sum_{i=0}^{n-1} 2^{ti} = \alpha + \beta \frac{2^{tn} - 1}{2^t - 1}  \quad (\alpha, \beta \in \Q).
\end{split}
\]
}
This gives us an exponential polynomial with only constant polynomials.

A set of the form $\{\, g_1 + \cdots + g_m \mid m \le M, g_i \in G \,\}$ for some $M \ge 0$ and a finitely generated subgroup $G \le \Q^\times$ is called a Bézivin set \cite{puch-smertnig24}.
One can show that $\N$ is not a Bézivin set.\footnote{This is a consequence of a theorem of Bézivin \cite[Th.~4]{Bezivin86}: if $\N$ were Bézivin, its generating series $\sum_{n=0}^\infty n x^n = \frac{x}{(1-x)^2}$ would have to have simple poles only, which is not the case.}
It is also easy to see that the output set of a finitely ambiguous weighted automaton is a Bézivin set.
Since $f(\Sigma^*) = \N$, the function $f$ cannot be recognised by a finitely ambiguous weighted automaton.
\end{example}

\section{Equivalence and Zeroness of CCRA}\label{sec:equivalence}
The two problems are defined as follows (for any classes of automata):
\begin{itemize}
 \item \emph{equivalence}: given two automata $\A$ and $\B$, decide if $\A(w) = \B(w)$ for all $w\in\Sigma^*$.
 \item \emph{zeroness}: given an automaton $\A$, decide if $\A(w) = 0$ for all $w$.
\end{itemize}
It is folklore that for (polynomially) weighted automata and CCRA the two problems are effectively interreducible. Indeed, to decide zeroness of $\A$ it suffices to check equivalence with $\B$ that outputs $0$ on all words. Conversely, to check equivalence of $\A$ and $\B$ one can check zeroness of $\A - \B$, which can be efficiently constructed for these models. Therefore we will deal only with zeroness.

For polynomially ambiguous weighted automata, even unrestricted weighted automata, we know that zeroness is in polynomial time~\cite{Schutzenberger61b} and in NC$^2$~\cite{Tzeng96}. Thus, we focus on the complexity of zeroness for CCRA. For the problem to make sense we need to introduce the size of the input CCRA. Given $\C = (Q, q_0, d, \Sigma, \delta, \mu, \nu)$, we say that its size is $|Q| + d + |\Sigma| + \max_p(|p|)$, where $p$ ranges over all polynomials and constants used in $\delta$, $\mu$ and $\nu$. We assume that polynomials are represented in the natural succinct form of arithmetic tree circuits, not as a list of all monomials.

Recall that in the update function one cannot use polynomials like $x^2$ because two copies of $x$ are needed. However, in some sense CCRA can evaluate any polynomial. For example there is a CCRA $\C$ such that $\C(1^n) = n^2$, simply by having two registers storing $n$ and defining the output function as their product. We can say that evaluating $x^2$ requires two copies of $x$. We generalise this observation to any polynomial. Given $x_1,\ldots,x_k$, we say that a polynomial $p(x_1,\ldots,x_k)$ is $d$-copyless if there exists a copyless polynomial $p'(\vec{x})$, where $\vec{x} = x_{1,1},\ldots,x_{1,k}, \ldots x_{d,1},\ldots,x_{d,k}$ ($d$ copies of every $x_i$) such that $p'(\vec{x}) = p(x_1,\ldots,x_k)$, substituting $x_{i,j} = x_j$ for all $1 \le i \le d$ and $1 \le j \le k$. In particular $1$-copyless is copyless.

We will use a standard and convenient lemma that allows us to turn formulas into polynomials. Note that we assume that formulas, like polynomials, are represented as tree circuits. By the size of the formula, denoted $|\varphi|$, we understand the size of the circuit. The proof can be found in Appendix~\ref{sec:additional-proofs}.
\begin{restatable}[]{lemma}{formulatopoly}\label{lem:formula_to_poly}
Let $\vec{x} = (x_1,\ldots,x_k)$ and let $\varphi(\vec{x})$ be a Boolean quantifier free formula. There exists a polynomial $p(\vec{x})$, of size polynomial in $|\varphi|$, such that for every $\vec{v} \in \set{0,1}^k$ we have: $p(\vec{v}) \in \set{0,1}$; and $p(\vec{v}) = 1$ if and only if $\varphi(\vec{v})$ evaluates to true.
Moreover, the polynomial $p(\vec{x})$ is $|\varphi|$-copyless.
\end{restatable}

We are ready to prove the main theorem.

\theorempspace*

\begin{proof}
Regarding the upper bound, by \cref{lemma:copyless_to_WA}, we know that a CCRA can be translated to a weighted automaton of exponential size. It is known that the equivalence problem for weighted automata is in NC$^2$~\cite{Tzeng96}. Since problems in NC$^2$ can be solved sequentially in polylogarithmic space~\cite{Ruzzo81}, this essentially yields a \textsc{PSpace} algorithm. One has to take care that the weighted automaton is not fully precomputed (as it would require too much space). A standard approach computing  the states and transitions on the fly solves this issue. See e.g.\ \cite[Section 6.1]{JeckerMP24} for a similar construction.

The rest of the proof is devoted to the matching \textsc{PSpace} lower bound.
We reduce from the validity problem for Quantified Boolean Formulas (QBF), which is known to be \textsc{PSpace}-complete \cite[Theorem 19.1]{papadimitriou2003computational}. One can assume the input is a formula of the form
\begin{align}\label{eq:qbf}
\psi =  \forall x_1 \exists y_1 \ldots \forall x_k \exists y_k \;\; \varphi(x_1,y_1,\ldots,x_k,y_k),
\end{align}
where $\varphi$ is quantifier-free. The variables $x_i$ and $y_i$ alternate, $x_i$ are quantified universally and $y_i$ are quantified existentially. For simplicity, we write $\vec{x} = (x_1,\ldots,x_k)$ and $\vec{y} = (y_1,\ldots,y_k)$. We write $\varphi(\vec{x},\vec{y})$ instead of $\varphi(x_1,y_1,\ldots,x_k,y_k)$.
Given $\vec{v} \in \set{0,1}^{2k}$ we denote by $\varphi(\vec{v})$ the truth value of the formula $\varphi$ with all variables evaluated according to $\vec{v}$.

For the reduction, we will need to go through many evaluations of $x_i$ and $y_i$ in a way that respects the quantifiers. It will be convenient to define these evaluations using auxiliary formulas. Let $\vec{x}'$ and $\vec{y}'$ be fresh copies of variables in $\vec{x}$ and $\vec{y}$ all of dimension $k$. We define three quantifier-free formulas: $\textsc{start}(\vec{x},\vec{y})$, $\textsc{next}(\vec{x},\vec{y},\vec{x}',\vec{y}')$ and $\textsc{end}(\vec{x},\vec{y})$, as follows.

\begin{align*}
 \textsc{start}(\vec{x},\vec{y}) = \bigwedge_{i = 1}^k \neg x_i, \hspace{1cm} \textsc{end}(\vec{x},\vec{y}) = \bigwedge_{i = 1}^k x_i.
\end{align*}
Note that $\textsc{start}$ and $\textsc{end}$ do not use $\vec{y}$, but in this form it will be easier to state the claim later explaining their purpose. We also define
\begin{align*}
 \textsc{next}(\vec{x},\vec{y},\vec{x}',\vec{y}') = \bigwedge_{i = 1}^k \hspace{0.5cm} \left( \neg x_i \wedge \bigwedge_{j = i+1}^k x_j \right) \implies \\
  \left( x_i' \wedge \bigwedge_{j = i+1}^k \neg x_j' \;\; \wedge \;\; \bigwedge_{j=1}^{i-1} (x_j \iff x_j') \wedge (y_j \iff y_j')  \right).
\end{align*}

To understand the formulas, it is easier to ignore the $\vec{y}$ and $\vec{y}'$ variables at first. Then these formulas essentially encode a binary counter with $k$ bits: $\textsc{start}$ encodes that all $x_i$ are $0$; $\textsc{end}$ encodes that all $x_i$ are $1$; and $\textsc{next}$ encodes that $\vec{x}'$ is $\vec{x}$ increased by $1$ in binary. The values of $\vec{y}$ can be guessed to anything in $\textsc{start}$ and $\textsc{end}$. In $\textsc{next}$ we keep consistently the guessed existential values for all unchanged universal variables.
The following lemma formally states the purpose of the formulas.

\begin{claim}\label{claim:formulas}
The formula $\psi$ in \cref{eq:qbf} is valid if and only if there exists a sequence $\vec{v}_1,\ldots,\vec{v}_n \in \set{0,1}^{2k}$ such that:
\begin{enumerate}
\item $\varphi(\vec{v}_i)$ is true for all $1 \le i \le n$;
 \item $\textsc{start}(\vec{v}_1)$ is true;
 \item $\textsc{next}(\vec{v}_i,\vec{v}_{i+1})$ is true for all $1 \le i \le n-1$;
 \item $\textsc{end}(\vec{v}_n)$ is true.
\end{enumerate}
\end{claim}

\begin{claimproof}
The formulas $\textsc{start}$, $\textsc{next}$ and $\textsc{end}$ are defined in such a way that they go through all possible evaluations of universal variables, guessing consistently the values for existential variables. The first condition guarantees that $\psi$ is valid.
\end{claimproof}

Thanks to \cref{claim:formulas} we will not need to differentiate between universal and existential variables. 
In the following we will implicitly use \cref{lem:formula_to_poly}.
To avoid additional notation we will write formula names for their corresponding polynomials.
Let $\ell = \max\{|\textsc{start}|, |\textsc{end}|, |\textsc{next}|, |\varphi|\}$, then all polynomials corresponding to these formulas are $\ell$-copyless.
To ease the notation we write
$$
\vec{z} = (x_1^1, \ldots, x_k^1, y_1^1, \ldots, y_k^1, \ldots, x_1^\ell, \ldots, x_k^\ell, y_1^\ell, \ldots, y_k^\ell)
$$ for $\ell$ identical copies of vectors of variables in $\vec{x}$ and $\vec{y}$. Note that identical copies occur on indices equal modulo $2k$ (this will be useful when defining the transitions). The number of copies will be sufficient to evaluate all polynomials corresponding to the formulas in a copyless manner. To emphasise this, we will write $\textsc{start}(\vec{z})$, $\textsc{end}(\vec{z})$ and $\textsc{next}(\vec{z})$.

We are ready to define the CCRA $\C = (Q, p_0, d, \Sigma, \delta, \mu, \nu)$, where: $Q = \set{\, p_i, q_i \mid 0 \le i \le 2k\,}$; $d = 8\ell k+1$; $\Sigma = \set{0,1,\#}$. We denote the $8\ell k+1$ variables as follows: $\vec{z}$, $\vec{z}'$, $\vec{z}''$, $\vec{z}^{\text{old}}$ and $s$. That is: four disjoint copies corresponding to $\ell$ copies of $\vec{x}$, $\vec{y}$ and one extra variable $s$. We denote the variables in the copies by $z_i$, $z_i'$, $z_i''$, $z_i^{\text{old}}$ for $1 \le i \le 2\ell k$. The initial function assigns the value $1$ to all variables. It will be important that $\mu(s) = 1$; for all other variables the initial value could be arbitrary.
The final function is defined by $\nu(x) =0$ for all $x \in Q \setminus \set{q_{0}}$ and
$\nu(q_{0}) =  s \cdot \textsc{end}(\vec{z}^{\text{old}})$.

Before we define the transitions we give an intuition on how the automaton works. We call a subword of length $2k$ a block. The automaton will read a sequence of blocks which correspond to consecutive evaluations $\vec{v}_i$ from \cref{claim:formulas} and store them in multiple copies of $\vec{x}$ and $\vec{y}$. After reading every block the automaton will check whether: $\textsc{next}$ holds with the previous block; and whether $\varphi$ holds on the current block. As an invariant, the register $s$ will have value $1$ if no error has been detected, and $0$ otherwise.

Most of the transitions will initialise some registers. Given a set of variables $Z$ and $b \in \set{0,1}$, we define the copyless polynomial map $P_{Z,b}$ as: $P_{Z,b}(z) = b$ for $z \in Z$ and $P_{Z,b}(z) = z$ otherwise. In words, the variables in $Z$ are initialised to $b$ and all others keep their previous value. We will use one type of sets $Z$, defined as follows: $Z_i = \set{\,z_j, z_j',z_j'' \mid j \equiv i \mod 2k\,}$. This will allow us to remember $3\ell$ copies at once.

Formally, we define the transitions as follows (see \cref{fig:CCRA_QBF} for the shape of the automaton without the register updates):
\begin{enumerate}
 \item\label{eq:guess1} $\delta(p_{i-1},b) = (p_{i},P_{Z_i,b})$ for all $1 \le i \le 2k$ and $b \in \set{0,1}$.
\item\label{eq:guess2} $\delta(q_{i-1},b) = (q_{i},P_{Z_i,b})$ for all $1 \le i \le 2k$ and $b \in \set{0,1}$.
 \item\label{eq:firsthash} $\delta(p_{2k},\#) = (q_0,Q_{0})$, where $Q_0$ resets all variables to $0$ except for: $\vec{z}^{\text{old}}$ where it puts the content of $\vec{z}''$, i.e., $Q_0(z_i^{\text{old}}) = z_i''$ for all $1 \le i \le 2\ell k$; and $Q_0(s) = s \cdot \textsc{start}(\vec{z}') \cdot \varphi(\vec{z})$.
 \item\label{eq:secondhash} $\delta(q_{2k},\#) = (q_0',R_{0})$, where $R_0$ resets all variables to $0$ except for: $\vec{z}^{\text{old}}$ where it puts the content of $\vec{z}''$, i.e., $Q_0(z_i^{\text{old}}) = z_i''$ for all $1 \le i \le 2\ell k$; and $Q_0(s) = s \cdot \textsc{next}(\vec{z}^{\text{old}},\vec{z}') \cdot \varphi(\vec{z})$.
\end{enumerate}
Note that all polynomials are copyless.

\begin{figure}
\centering
\begin{tikzpicture}
\node[state] (p0) {$p_0$};
\node[state, right = 1cm of p0] (p1) {$p_1$};
\node[state, right = 1cm of p1] (p2) {$p_2$};

\node[state, right = 1cm of p2] (q0) {$q_0$};
\node[state, right = 1cm of q0] (q1) {$q_1$};
\node[state, right = 1cm of q1] (q2) {$q_2$};
   
\path
(p0) edge[->,bend right] node[below] {$0$} (p1)
(p1) edge[->,bend right] node[below] {$0$} (p2)
(p0) edge[->,bend left] node[above] {$1$} (p1)
(p1) edge[->,bend left] node[above] {$1$} (p2)
(p2) edge[->] node[above] {$\#$} (q0)
(q0) edge[->,bend right] node[below] {$0$} (q1)
(q1) edge[->,bend right] node[below] {$0$} (q2)
(q0) edge[->,bend left] node[above] {$1$} (q1)
(q1) edge[->,bend left] node[above] {$1$} (q2)
(q2) edge[->,bend right = 60] node[above] {$\#$} (q0)
(q0.north) edge[->] ++ (0,0.5)
(p0.west) edge[<-] ++ (-0.5, 0)
;
\end{tikzpicture}
\caption{Example for $k = 1$. The state $q_0$ has an outgoing edge as it is the only one that has a possibly nonzero output.}\label{fig:CCRA_QBF}
\end{figure}
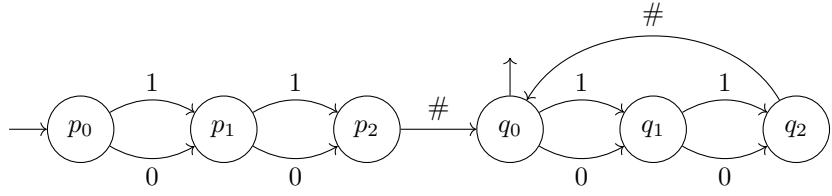

The proof that the reduction works follows essentially from \cref{claim:formulas}. The transitions in \cref{eq:guess1} and \cref{eq:guess2} guess the evaluations $\vec{v}_i$. These are stored in three copies: $\vec{z}$, $\vec{z}'$, $\vec{z}''$. The remaining two transitions verify the correctness of these evaluations, i.e., whether they satisfy the conditions in \cref{claim:formulas}. Note that as an invariant these transitions keep in $\vec{v}^{\text{old}}$ the previous valuation. In both \cref{eq:firsthash}, \cref{eq:secondhash} we check whether $\varphi(\vec{v}_i)$ holds. Additionally, in \cref{eq:firsthash} we check whether $\textsc{start}(\vec{v}_1)$ is true; and in \cref{eq:secondhash} we check whether $\textsc{next}(\vec{v}_{i-1},\vec{v}_{i})$ is true. All checks are multiplied into the register $s$, which becomes $0$ if any error occurs, and remains $1$ otherwise. Finally, the output function guarantees that a nonzero value can be output only if $\textsc{end}(\vec{v}_n)$ holds.
\end{proof}



\bibliography{lipics-v2021-sample-article}

\begin{thebibliography}{10}

\bibitem{AlurDDRY13}
Rajeev Alur, Loris D'Antoni, Jyotirmoy~V. Deshmukh, Mukund Raghothaman, and
  Yifei Yuan.
\newblock {R}egular {F}unctions and {C}ost {R}egister {A}utomata.
\newblock In {\em 28th Annual {ACM/IEEE} Symposium on Logic in Computer
  Science, {LICS} 2013, New Orleans, LA, USA, June 25-28, 2013}, pages 13--22.
  {IEEE} Computer Society, 2013.
\newblock \href {https://doi.org/10.1109/LICS.2013.65}
  {\path{doi:10.1109/LICS.2013.65}}.

\bibitem{BarloyFLM22}
Corentin Barloy, Nathana{\"{e}}l Fijalkow, Nathan Lhote, and Filip Mazowiecki.
\newblock A robust class of linear recurrence sequences.
\newblock {\em Inf. Comput.}, 289(Part):104964, 2022.
\newblock \href {https://doi.org/10.1016/J.IC.2022.104964}
  {\path{doi:10.1016/J.IC.2022.104964}}.

\bibitem{bell2021noncommutative}
Jason Bell and Daniel Smertnig.
\newblock Noncommutative rational {P}\'olya series.
\newblock {\em Selecta Math. (N.S.)}, 27(3):Paper No. 34, 34, 2021.
\newblock \href {https://doi.org/10.1007/s00029-021-00629-2}
  {\path{doi:10.1007/s00029-021-00629-2}}.

\bibitem{BellS23}
Jason~P. Bell and Daniel Smertnig.
\newblock Computing the linear hull: Deciding deterministic? and unambiguous?
  for weighted automata over fields.
\newblock In {\em 38th Annual {ACM/IEEE} Symposium on Logic in Computer
  Science, {LICS} 2023, Boston, MA, USA, June 26-29, 2023}, pages 1--13.
  {IEEE}, 2023.
\newblock \href {https://doi.org/10.1109/LICS56636.2023.10175691}
  {\path{doi:10.1109/LICS56636.2023.10175691}}.

\bibitem{BenaliouaLR24}
Yahia~Idriss Benalioua, Nathan Lhote, and Pierre{-}Alain Reynier.
\newblock {M}inimizing {C}ost {R}egister {A}utomata over a {F}ield.
\newblock In Rastislav Kr{\'{a}}lovic and Anton{\'{\i}}n Kucera, editors, {\em
  49th International Symposium on Mathematical Foundations of Computer Science,
  {MFCS} 2024, August 26-30, 2024, Bratislava, Slovakia}, volume 306 of {\em
  LIPIcs}, pages 23:1--23:15. Schloss Dagstuhl - Leibniz-Zentrum f{\"{u}}r
  Informatik, 2024.
\newblock \href {https://doi.org/10.4230/LIPICS.MFCS.2024.23}
  {\path{doi:10.4230/LIPICS.MFCS.2024.23}}.

\bibitem{berstel-reutenauer11}
Jean Berstel and Christophe Reutenauer.
\newblock {\em Noncommutative {R}ational {S}eries with {A}pplications}, volume
  137 of {\em Encyclopedia of Mathematics and its Applications}.
\newblock Cambridge University Press, Cambridge, 2011.

\bibitem{Bezivin86}
Jean-Paul B\'ezivin.
\newblock Sur un th\'eor\`eme de {G}. {P}\'olya.
\newblock {\em J. Reine Angew. Math.}, 364:60--68, 1986.
\newblock \href {https://doi.org/10.1515/crll.1986.364.60}
  {\path{doi:10.1515/crll.1986.364.60}}.

\bibitem{CadilhacMPPS24}
Micha{\"{e}}l Cadilhac, Filip Mazowiecki, Charles Paperman, Michal Pilipczuk,
  and G{\'{e}}raud S{\'{e}}nizergues.
\newblock {O}n {P}olynomial {R}ecursive {S}equences.
\newblock {\em Theory Comput. Syst.}, 68(4):593--614, 2024.
\newblock \href {https://doi.org/10.1007/S00224-021-10046-9}
  {\path{doi:10.1007/S00224-021-10046-9}}.

\bibitem{ChattopadhyayMM21}
Agnishom Chattopadhyay, Filip Mazowiecki, Anca Muscholl, and Cristian Riveros.
\newblock Pumping lemmas for weighted automata.
\newblock {\em Log. Methods Comput. Sci.}, 17(3), 2021.
\newblock \href {https://doi.org/10.46298/LMCS-17(3:7)2021}
  {\path{doi:10.46298/LMCS-17(3:7)2021}}.

\bibitem{CzerwinskiLMPW22}
Wojciech Czerwinski, Engel Lefaucheux, Filip Mazowiecki, David Purser, and
  Markus~A. Whiteland.
\newblock The boundedness and zero isolation problems for weighted automata
  over nonnegative rationals.
\newblock In Christel Baier and Dana Fisman, editors, {\em {LICS} '22: 37th
  Annual {ACM/IEEE} Symposium on Logic in Computer Science, Haifa, Israel,
  August 2 - 5, 2022}, pages 15:1--15:13. {ACM}, 2022.
\newblock \href {https://doi.org/10.1145/3531130.3533336}
  {\path{doi:10.1145/3531130.3533336}}.

\bibitem{DaviaudJLMP021}
Laure Daviaud, Marcin Jurdzinski, Ranko Lazic, Filip Mazowiecki, Guillermo~A.
  P{\'{e}}rez, and James Worrell.
\newblock When are emptiness and containment decidable for probabilistic
  automata?
\newblock {\em J. Comput. Syst. Sci.}, 119:78--96, 2021.
\newblock \href {https://doi.org/10.1016/J.JCSS.2021.01.006}
  {\path{doi:10.1016/J.JCSS.2021.01.006}}.

\bibitem{droste2009handbook}
Manfred Droste, Werner Kuich, and Heiko Vogler.
\newblock {\em Handbook of {W}eighted {A}utomata}.
\newblock Springer Science \& Business Media, 2009.

\bibitem{everest2003recurrence}
Graham Everest, Alfred~Jacobus Van Der~Poorten, Igor Shparlinski, Thomas Ward,
  et~al.
\newblock {\em {R}ecurrence {S}equences}, volume 104.
\newblock American Mathematical Society Providence, RI, 2003.

\bibitem{FijalkowRW22}
Nathana{\"{e}}l Fijalkow, Cristian Riveros, and James Worrell.
\newblock Probabilistic automata of bounded ambiguity.
\newblock {\em Inf. Comput.}, 282:104648, 2022.
\newblock \href {https://doi.org/10.1016/J.IC.2020.104648}
  {\path{doi:10.1016/J.IC.2020.104648}}.

\bibitem{halava2005skolem}
Vesa Halava, Tero Harju, Mika Hirvensalo, and Juhani Karhum{\"a}ki.
\newblock {S}kolem’s {P}roblem --- {O}n the {B}order between {D}ecidability
  and {U}ndecidability.
\newblock {\em TUCS Technical Reports}, 683, 2005.

\bibitem{Hashiguchi88}
Kosaburo Hashiguchi.
\newblock {A}lgorithms for {D}etermining {R}elative {S}tar {H}eight and {S}tar
  {H}eight.
\newblock {\em Inf. Comput.}, 78(2):124--169, 1988.
\newblock \href {https://doi.org/10.1016/0890-5401(88)90033-8}
  {\path{doi:10.1016/0890-5401(88)90033-8}}.

\bibitem{JeckerMP24}
Isma{\"{e}}l Jecker, Filip Mazowiecki, and David Purser.
\newblock {D}eterminisation and {U}nambiguisation of {P}olynomially-{A}mbiguous
  {R}ational {W}eighted {A}utomata.
\newblock In Pawel Sobocinski, Ugo~Dal Lago, and Javier Esparza, editors, {\em
  Proceedings of the 39th Annual {ACM/IEEE} Symposium on Logic in Computer
  Science, {LICS} 2024, Tallinn, Estonia, July 8-11, 2024}, pages 46:1--46:13.
  {ACM}, 2024.
\newblock \href {https://doi.org/10.1145/3661814.3662073}
  {\path{doi:10.1145/3661814.3662073}}.

\bibitem{Kostolanyi23}
Peter Kostol{\'{a}}nyi.
\newblock {P}olynomially {A}mbiguous {U}nary {W}eighted {A}utomata over
  {F}ields.
\newblock {\em Theory Comput. Syst.}, 67(2):291--309, 2023.
\newblock \href {https://doi.org/10.1007/S00224-022-10107-7}
  {\path{doi:10.1007/S00224-022-10107-7}}.

\bibitem{MazowieckiR15}
Filip Mazowiecki and Cristian Riveros.
\newblock {M}aximal {P}artition {L}ogic: {T}owards a {L}ogical
  {C}haracterization of {C}opyless {C}ost {R}egister {A}utomata.
\newblock In Stephan Kreutzer, editor, {\em 24th {EACSL} Annual Conference on
  Computer Science Logic, {CSL} 2015, September 7-10, 2015, Berlin, Germany},
  volume~41 of {\em LIPIcs}, pages 144--159. Schloss Dagstuhl - Leibniz-Zentrum
  f{\"{u}}r Informatik, 2015.
\newblock \href {https://doi.org/10.4230/LIPICS.CSL.2015.144}
  {\path{doi:10.4230/LIPICS.CSL.2015.144}}.

\bibitem{MazowieckiR19}
Filip Mazowiecki and Cristian Riveros.
\newblock {C}opyless cost-register automata: {S}tructure, expressiveness, and
  closure properties.
\newblock {\em J. Comput. Syst. Sci.}, 100:1--29, 2019.
\newblock \href {https://doi.org/10.1016/J.JCSS.2018.07.002}
  {\path{doi:10.1016/J.JCSS.2018.07.002}}.

\bibitem{OuaknineW15}
Jo{\"{e}}l Ouaknine and James Worrell.
\newblock On linear recurrence sequences and loop termination.
\newblock {\em {ACM} {SIGLOG} News}, 2(2):4--13, 2015.
\newblock \href {https://doi.org/10.1145/2766189.2766191}
  {\path{doi:10.1145/2766189.2766191}}.

\bibitem{papadimitriou2003computational}
Christos~H. Papadimitriou.
\newblock Computational complexity.
\newblock In {\em Encyclopedia of Computer Science}, pages 260--265. 2003.

\bibitem{paz71}
Azaria Paz.
\newblock {\em {I}ntroduction to {P}robabilistic {A}utomata}.
\newblock Academic Press, 1971.

\bibitem{puch-smertnig24}
Antoni Puch and Daniel Smertnig.
\newblock {F}actoring through monomial representations: {A}rithmetic
  characterizations and ambiguity of weighted automata.
\newblock 2024.
\newblock Preprint.
\newblock \href {https://arxiv.org/abs/2410.03444} {\path{arXiv:2410.03444}}.

\bibitem{Reutenauer79}
Christophe Reutenauer.
\newblock On {P}olya series in noncommuting variables.
\newblock In Lothar Budach, editor, {\em Fundamentals of Computation Theory,
  {FCT} 1979, Proceedings of the Conference on Algebraic, Arthmetic, and
  Categorial Methods in Computation Theory, Berlin/Wendisch-Rietz, Germany,
  September 17-21, 1979}, pages 391--396. Akademie-Verlag, Berlin, 1979.

\bibitem{Ruzzo81}
Walter~L. Ruzzo.
\newblock {O}n {U}niform {C}ircuit {C}omplexity.
\newblock {\em J. Comput. Syst. Sci.}, 22(3):365--383, 1981.
\newblock \href {https://doi.org/10.1016/0022-0000(81)90038-6}
  {\path{doi:10.1016/0022-0000(81)90038-6}}.

\bibitem{Schutzenberger61b}
Marcel~Paul Sch{\"{u}}tzenberger.
\newblock {O}n the {D}efinition of a {F}amily of {A}utomata.
\newblock {\em Inf. Control.}, 4(2-3):245--270, 1961.
\newblock \href {https://doi.org/10.1016/S0019-9958(61)80020-X}
  {\path{doi:10.1016/S0019-9958(61)80020-X}}.

\bibitem{Tzeng96}
Wen{-}Guey Tzeng.
\newblock {O}n {P}ath {E}quivalence of {N}ondeterministic {F}inite {A}utomata.
\newblock {\em Inf. Process. Lett.}, 58(1):43--46, 1996.
\newblock \href {https://doi.org/10.1016/0020-0190(96)00039-7}
  {\path{doi:10.1016/0020-0190(96)00039-7}}.

\end{thebibliography}
\clearpage
\appendix

\section{Exponential Polynomials and LRS} \label{subsec:appendix-exppoly}

In this appendix we discuss exponential polynomials, exponential polynomial sequences and their relation to rational functions and LRS in somewhat more detail.
The material and treatment are essentially standard, see also \cite[Chapter~6]{berstel-reutenauer11}\cite{everest2003recurrence}\cite[Proposition 2.11]{halava2005skolem}, but a few points are somewhat subtle in positive characteristic.

Every EPS satisfies a linear recurrence and is hence an LRS.
In the converse direction, in characteristic $0$, if $(a_n)_n$ is an LRS, then there exists some $n_0$ such that $a_n$ for $n \ge n_0$ coincides with a (unique) EPS.
So, in characteristic $0$, there is a close relation between LRS and EPS.

In characteristic $p>0$, it is no longer true that every LRS coincides with an EPS (even for large enough $n$), because the map $n \mapsto \overline{\binom{n+j-1}{j-1}}$ for $j > p$ does not factor through a function $\bF_p \to \bF_p$, and hence cannot be represented by a polynomial $q \in \bF_p[x]$.
However, if an LRS can be expressed as an EPS, the EPS is still unique, with the caveat
that the EPS has multiple representations by exponential polynomials, leading us to consider the minimal degree representation, that is, that of degree at most $p-1$ (which is unique).
We now discuss this in detail.

Throughout the section, let $K$ be a field and $\overline{K}$ its algebraic closure.
If $K=\Q$, this is the field of all algebraic numbers $\overline{\Q}$, a subfield of the complex numbers.
The characteristic of $K$ is either $0$ or a prime number $p$.
If it is $0$, then $\Q$ embeds uniquely into $K$ as prime field, and we can assume $\Q \subseteq K$.
Similarly, if $p > 0$, then $\bF_p \subseteq K$.

A sequence $(a_n)_n$ is an LRS if and only if the (formal) generating series $F = \sum_{n=0}^\infty a_n x^n \in K\llbracket x \rrbracket$ is a rational function.
Using the existence and uniqueness of partial fraction decompositions of rational functions over the algebraically closed field $\overline{K}$, one obtains the following well-known result.

\begin{theorem} \label{t:lrs-repr}
    Let $(a_n)_{n}$ be an LRS over $K$.
    Then there exist $l \ge 0$, pairwise distinct $\lambda_1$, \dots,~$\lambda_l \in \overline{K}^\times$, and for each $1 \le i \le l$ natural numbers $k_i \ge 0$ and coefficients $\alpha_{i,1}$, \dots, $\alpha_{i,k_i} \in \overline{K}$ with $\alpha_{i,k_i} \ne 0$ such that
    \begin{equation} \label{eq:pfd}
    a_n = \sum_{i=1}^l \sum_{j=1}^{k_i} \alpha_{i,j} \binom{n+j-1}{j-1} \lambda_i^{n} \qquad\text{for all sufficiently large $n$.}
    \end{equation}
    The elements $\lambda_i$ and $\alpha_{i,j}$ are uniquely determined by the sequence $(a_n)_n$.
    The set $\{ \lambda_1, \dots, \lambda_l \}$ is the set of nonzero characteristic roots of $(a_n)_n$.
\end{theorem}

We also recall (but do not need) that \cref{eq:pfd} holds for \emph{all} $n \ge 0$, that is, not only for sufficiently large $n$, if and only if all characteristic roots of $(a_n)_n$ are nonzero.
A converse to \cref{t:lrs-repr} also holds: every sequence expressed as in \cref{eq:pfd} is an LRS over the algebraically closed $\overline{K}$.

\begin{example}
    The sequence $(a_n)_n=n$, whose only characteristic root is $1$, can be represented as $a_n = n \cdot 1^n = \big(\binom{n+1}{1} - \binom{n}{0}\big) \cdot 1^n$.
    For the Fibonacci sequence one obtains the well-known representation $F_n = \tfrac{1}{\sqrt{5}} \varphi^n - \tfrac{1}{\sqrt{5}} \psi^n$.
    Note that $\varphi$, $\psi \in \overline{\Q} \setminus \Q$.
    While the Fibonacci sequence cannot be recognised by a polynomially ambiguous weighted automaton with weights in $\Q$, this formula shows that there is such an automaton with weights in the quadratic field $\Q(\varphi)$.
\end{example}

Before proving \cref{t:lrs-repr}, we recall one more lemma (and in particular, that it also holds in positive characteristic).

\begin{lemma}
    For all $\alpha \in K$,
    \[
    \frac{1}{(1-\alpha x)^k} = \sum_{n = 0}^\infty \binom{n+k-1}{k-1} \alpha^n x^n \in K\llbracket x \rrbracket.
    \]
\end{lemma}
    
\begin{proof}
    First suppose $\alpha=1$.
    For $K=\Z$, this is well-known and easily derived from the geometric series $(1-x)^{-1} = \sum_{n=0}^\infty x^n$ by formal differentiation (or a combinatorial argument).
    The ring homomorphism $\Z \to K$, $1 \mapsto 1_K$ extends coefficient-wise to a ring homomorphism $\Z\llbracket x \rrbracket \to K\llbracket x \rrbracket$.
    Applying the homomorphism to the identity
    \[
    (1-x)^k \sum_{n=0}^{\infty} \binom{n+k-1}{k-1} x^n = 1,
    \]
    and dividing by $(1_K - x)^k$ in $K\llbracket x \rrbracket$ proves the claim for $\alpha=1$.
    For arbitrary $\alpha \in K$, it follows by substituting $\alpha x$ for $x$.
\end{proof}
    
\begin{proof}[Proof of \cref{t:lrs-repr}]
    We may without restriction assume that $K$ is algebraically closed.
    Because $(a_n)_n$ satisfies an LRS, its generating function $F(x)=\sum_{n=0}^\infty a_n x^n \in K\llbracket x \rrbracket$ is rational.
    Thus, there exist coprime polynomials $p$,~$q \in K[x]$ such that $F = p/q$ with $q=(x-\lambda_1^{-1})^{k_1} \cdots (x- \lambda_l^{-1})^{k_l}$, where $\lambda_1$, \dots,~$\lambda_r \in \overline{K}^\times$ are pairwise distinct and $k_i \ge 1$.
    By partial fraction decomposition, there are uniquely determined $\alpha_{i,j} \in \overline{K}$ and a uniquely determined polynomial $r \in K[x]$ such that
    \[
    F = r + \sum_{i=1}^l \sum_{j=1}^{k_i} \frac{\alpha_{i,j}}{(1 - \lambda_i x)^j} = r + \sum_{i=1}^l \sum_{j=1}^{k_i} \alpha_{i,j} \sum_{n = 0}^\infty \binom{n+j-1}{j-1} \lambda_i^n x^n,
    \]
    which shows existence of the claimed representation for all $n > \deg(r)$.
    The uniqueness follows from the uniqueness of the partial fraction decomposition.
\end{proof}

To go further than \cref{t:lrs-repr}, we need to distinguish according to the characteristic of $K$.
If $\chr K = 0$, then $\binom{n+j-1}{j-1} = \frac{(n+1)\cdots (n+j-1)}{j!}$ allows us to view the binomial coefficients in \cref{eq:pfd} as polynomial functions in $n$.
Expanding shows that, for sufficiently large $n$, the sequence $(a_n)_n$ can be represented by a (uniquely determined) exponential polynomial \cite[Ch.~6.2]{berstel-reutenauer11}.

If $\chr K = p > 0$, then there neither needs to exist a representation by an exponential polynomial, nor need this representation be unique if it exists.

\begin{example} \label{exm:char2-not-exppoly}
    The triangular numbers $T_n \coloneqq \binom{n+1}{2} = \sum_{k=1}^n k$ satisfy the linear recurrence relation $T_{n+3} = 3T_{n+2} - 3T_{n+1} + T_n$ with $T_0=0$, $T_1=1$, $T_2=3$.
    The characteristic polynomial is $x^3 - 3x^2 + 3x - 1 = (x-1)^3$.
    In the normal form of \cref{t:lrs-repr},
    \[
    T_n = \binom{n+1}{2} \cdot 1^n = \bigg( \binom{n+2}{2} - \binom{n+1}{1} \bigg) \cdot 1^n
    \]
    In particular, in characteristic $0$, the sequence $T_n$, whose elements are $0$,~$1$, $3$, $6$, $10$, $15$, $21$,~$28$,~\dots, can be expressed using a polynomial, as $T_n = \tfrac{n^2 + n}{2}$.
    
    Now consider $K=\bF_2$. Reducing modulo $2$, the sequence $\overline{T_n} \in \bF_2$ still satisfies the same linear recurrence relation and $\overline{T_n} = \overline{\tfrac{n^2+n}{2}} = \overline{0}, \overline{1}, \overline{1}, \overline{0}, \overline{0}, \overline{1}, \overline{1}, \overline{0}$, etc.
    However, this sequence is not induced from a polynomial function $\bF_2 \to \bF_2$: indeed $(\overline{T_n})_n$ is not $2$-periodic. But, if there were a function $f\colon \bF_2 \to \bF_2$ such that $\overline{T_n}=f(\overline{n})$, then $\overline{T_n}$ would have to be $2$-periodic.
    We see that, in positive characteristic, not every LRS can be represented by an exponential polynomial (even for large enough $n$).
\end{example}

As the issues of uniqueness of exponential polynomial representations in positive characteristic are somewhat relevant in the present paper (in particular, in Appendix~\ref{ssec:ccra-positive-char}), we now discuss them in more detail.

A \emph{polynomial} $q=q(x) \in K[x]$ is a formal expression $q = \sum_{i=0}^k \alpha_i x^k$ with $\alpha_i \in K$.
Polynomials are multiplied and added according to the usual rules, by $K$-linear extension of $x^k \cdot x^l = x^{k+l}$.
Every polynomial $q$ induces a function $\overline{q}\colon K \to K$, $\lambda \mapsto q(\lambda)=\sum_{i=0}^k \alpha_i \lambda^k$, and every function of such a form is called a \emph{polynomial function}.
If the field $K$ is infinite (in particular if $\chr K=0$), then $q=q'$ for polynomials $q$,~$q' \in K[x]$ if and only if $\overline{q}=\overline{q'}$, that is, if and only if $q(\lambda)=q'(\lambda)$ for all $\lambda \in K$.
In this case, there is no need to carefully distinguish polynomials from polynomial functions.

\begin{example}
If $K$ is finite, then different polynomials may induce the same polynomial functions.
Indeed, there are infinitely many polynomials but only finitely many functions $K \to K$.
For instance, if $K=\bF_p$, then $x^p+1 \ne x + 1$ as polynomials, but $\lambda^p + 1 = \lambda +1$ for all $\lambda \in \bF_p$, so these two polynomials induce the same function $\bF_p \to \bF_p$.
\end{example}

As the previous example reminds us, in positive characteristic, we need to carefully distinguish between polynomials and polynomial functions.
In particular, for polynomial functions there is no canonical notion of the $i$-th coefficient, as the example shows.

In any characteristic, since there is a unique ring homomorphism $\Z \to K$, $n \mapsto n \cdot 1_K$, it also makes sense to evaluate polynomials at integers, and we can think of $q \in K[x]$ as inducing a sequence $\overline{q}\colon \N \to K$, $n \mapsto q(n \cdot 1_K)$ (by slight abuse of notation, overloading the notation $\overline{q}$).

We now extend these considerations from polynomials to exponential polynomials.
An \emph{exponential polynomial} is a formal expression $q(x)=\sum_{i=1}^k q_i(x) \lambda_i^x$ with $q_i \in K[x]$ and pairwise distinct $\lambda_i \in K^\times$.
The $\lambda_i$ for which $q_i\ne 0$ are the \emph{exponential bases} of $q$.
The \emph{degree} of $q$ is $\deg(q)\coloneqq \max\{\, \deg(q_i) \mid 1 \le i \le k\,\}$.
Exponential polynomials are again added and multiplied in the usual way, by $K$-linearly extending $(x^m \lambda_i^x) \cdot (x^n \lambda_j^x) = x^{m+n} (\lambda_i \lambda_j)^x$.\footnote{This can be made rigorous by considering exponential polynomials as elements of the group algebra $K[x][\Lambda]$ of the group $\Lambda=K^\times$ over the polynomial ring $K[x]$.
See \cite[Ch.~6]{berstel-reutenauer11}.}
An exponential polynomial $q$ induces a sequence $\overline{q}\colon \N \to K$, defined by $\overline{q}(n) = \sum_{i=1}^k q_i(n) \lambda_i^n$.
We call any sequence arising in such a way an \emph{exponential polynomial sequence \textup(EPS\textup)}.

The question to what degree the induced sequence determines the polynomials $q_i$ and the exponential bases $\lambda_i$ is answered by the following lemma.

\begin{lemma} \label{l:unique-exppoly}
    Let $p_1$, \dots,~$p_k$, $q_1$, \dots,~$q_l \in K[x]$. Let $\lambda_1$, \dots,~$\lambda_k \in K^\times$ be pairwise distinct, and let similarly $\mu_1$, \dots,~$\mu_l \in K^\times$ be pairwise distinct.
    Suppose
    \[
    \sum_{i=1}^k p_i(n) \lambda_i^n = \sum_{j=1}^l q_j(n) \mu_j^n \qquad\text{for all $n \in \N$}.
    \]
    Assume also that, for each $i$ and $j$ we have $p_i(\N) \ne \{0\}$ and $q_j(\N) \ne \{0\}$.
    Then, up to re-indexing, we have $k=l$, $\lambda_i=\mu_i$ for all $1 \le i \le k$, and $p_i(n) = q_i(n)$ for all $n \in \N$.
\end{lemma}

\begin{proof}
    Let $q \in K[x]$.
    It suffices to show that there exist $\alpha_1$, \dots,~$\alpha_m \in K$ such that $q(n) = \sum_{j=1}^m \alpha_j \binom{n+j-1}{j-1}$ for all $n \in \Z$.
    Then the claim follows from the (stronger) uniqueness statement of \cref{t:lrs-repr}.

    To show the desired expression for $q(n)$, it in turn suffices to show that $x^j \in \Z[x]$ is a $\Z$-linear combination of polynomials of the form $\binom{x+i-1}{i-1}$.
    For $j=0$, this is true because $x^0 = 1 = \binom{x}{0}$.
    For $j \ge 1$, note that $x^j - j! \binom{x+j}{j} = x^j - (x+1) \cdots (x+j)$ is a polynomial of degree strictly less than $j$, so the claim follows by induction.
\end{proof}

If $\chr K=0$, the uniqueness in the previous lemma implies that $p_i=q_i$ and so an exponential polynomial is uniquely determined by its induced EPS.

If $\chr K = p > 0$, then this is not the case.
In this case $p_i(n)=q_i(n)$ for all $n \in \N$ if and only if $p_i(n)=q_i(n)$ for $n \in \{0,\dots,p-1\}$.
This is the case if and only if the polynomial $p_i - q_i$ is divisible by $x(x-1) \cdots (x-p+1) = x^p - x$.
It follows that, when representing an EPS $(a_n)_n$ using an exponential polynomial $\sum_{i=1}^k p_i(x) \lambda_i^x$, we can always find a representation with $\deg(p_i) < p$.
In this case, the polynomials $p_i$ are uniquely determined by $(a_n)_n$, and we call the resulting exponential polynomial the \emph{exponential polynomial of minimal degree} representing $(a_n)_n$.

It is important to note that, independent of the characteristic of the field, the exponential bases appearing in a representation as in \cref{t:lrs-repr} or in an exponential polynomial representation (if it exists) are unique.
Further, every EPS is an LRS, and the exponential bases of the EPS are precisely the nonzero characteristic roots of the LRS.

\begin{definition}
    \label{definition:CoeffSum}
    If $q = \sum_{i=1}^m \sum_{j=0}^{l} \alpha_{i, j} x^j \lambda_i^x$ is an exponential polynomial, then $S_k(q)\coloneqq \sum_{i=1}^m \alpha_{i, k}$ is its \emph{sum of $k$-degree coefficients}.
\end{definition}

Thus, in $S_k(q)$, we are summing all coefficients next to some $x^k$, across all exponential bases.
The following straightforward observation on the behaviour of $S_k(q)$ on products will be useful in one of the proofs below.

\begin{lemma} \label{l:kcoeffsum-product}
    If $q$ and $q'$ are exponential polynomials, then $S_k(qq') =\sum_{j=0}^{k} S_j(q)S_{k-j}(q')$.
\end{lemma}

\begin{proof}
    Suppose
    \[
    q = \sum_{\lambda \in K^\times} \sum_{i\ge 0} \alpha_{\lambda,i} x^i \lambda^x \quad\text{and}\quad q'=\sum_{\lambda \in K^\times} \sum_{i\ge 0} \alpha'_{\lambda,i} x^i \lambda^x.
    \]
    (It is notationally convenient to allow formally infinite sums; but in each case there are only finitely many nonzero terms.)
    Then $S_j(q) = \sum_{\lambda \in K^\times} \alpha_{\lambda,j}$, and analogously for $q'$.
    Now
    \[
    qq' = \sum_{\lambda, \mu \in K^\times} \sum_{k \ge 0} \Big(\sum_{j=0}^k \alpha_{\mu,j} \alpha'_{\lambda\mu^{-1},k-j} \Big) x^{k} \lambda^x.
    \]
This shows
    \[
    \begin{split}
    S_k(qq') &= \sum_{\lambda, \mu \in K^\times} \sum_{j=0}^k \alpha_{\mu,j} \alpha'_{\lambda \mu^{-1},k-j} = \sum_{j=0}^k \Big(\sum_{\mu \in K^\times} \alpha_{\mu,j} \Big) \Big( \sum_{\mu \in K^\times} \alpha'_{\mu,k-j} \Big)\\
    &= \sum_{j=0}^k S_j(q) S_{k-j}(q'),
    \end{split}
    \]
    where the second to last step uses $K^\times = \{\, \lambda \mu^{-1} \mid \mu \in K^\times \,\}$.
\end{proof}

In Appendix~\ref{ssec:ccra-positive-char}, we also consider sums of $k$-degree coefficients of EPS $(a_n)_n$ in positive characteristic.
Now there are several exponential polynomials representing the EPS, and the sums of $k$-degree coefficients depend on the particular representation, not just on the sequence $(a_n)_n$.

\begin{example}
    Over $K=\bF_p$, we have $\overline{n}^p + \overline{n} + \overline{1} = \overline{2}\, \overline{n} + \overline{1}$.
    However, $q=x^p + x + \overline{1}$ has $S_p(q)=S_1(q)=S_0(q)=\overline{1}$ and $S_k(q)=\overline{0}$ for all other $k$.
    By contrast, the polynomial $q'=\overline{2} x + \overline{1}$ has $S_1(q)=\overline{2}$, $S_0(q)=\overline{1}$ and $S_k(q) = \overline{0}$ for all other $k$.
\end{example}

However, the obstruction in the example, that terms of the form $x^{p}$ can be replaced by $x$ without changing the induced function, is the only one.
More formally, we still have the following.

\begin{lemma} \label{l:charp-coeff-sums}
Suppose $\chr K = p > 0$.
Let $q$,~$q'$ be two exponential polynomials inducing the same EPS, that is, with $q(n)=q'(n)$ for all $n \in \N$.
Then $S_0(q)=S_0(q')$ and for every $r \in \{1,\dots,p-1\}$,
\[
\sum_{k\ge 0} {S_{k(p-1)+r}(q)} = \sum_{k\ge 0} {S_{k(p-1)+r}(q')}.
\]
\end{lemma}

While the sums in the lemma are formally infinite (for notational convenience), they only involve finitely many nonzero terms.
Applied to the minimal degree exponential polynomial $q$ representing an EPS $(a_n)_n$, it follows that $S_r(q) = \sum_{k\ge 0} S_{k(p-1)+r}(q')$ for every other exponential polynomial $q'$ representing $(a_n)_n$.

\begin{proof}[Proof of \cref{l:charp-coeff-sums}]
    Write $q = \sum_{\lambda \in K^\times} q_\lambda \lambda^x$ and $q' = \sum_{\lambda \in K^\times} q_\lambda' \lambda^x$ with polynomials $q_\lambda$, $q_\lambda'$ (only finitely many of which are nonzero).
    It will suffice to show the claimed formula for all pairs $q_\lambda$ and $q_\lambda'$ in place of $q$ and $q'$.

    Fix $\lambda \in K^\times$.
    Since $q-q'$ vanishes on all of $\N$, by \cref{l:unique-exppoly}, we must have $(q_\lambda - q_\lambda')(n) = 0$ for all $n \in \N$.
    This means that the polynomial $q_\lambda - q_\lambda'$ is divisible by $x(x-1)\cdots (x-p+1) = x^p - x$, that is, there exists a polynomial $h$ such that $q_\lambda - q_{\lambda'} = (x^p - x)h$.
    In particular $S_0(q_\lambda - q_{\lambda'}) = 0$, and so $S_0(q_\lambda)=S_0(q_{\lambda'})$.
    For $r \ge 1$, using \cref{l:kcoeffsum-product},
    \[
    \begin{split}
    \sum_{k\ge 0} {S_{k(p-1)+r}(q_\lambda - q_\lambda')} &= \sum_{k\ge 0} {S_{k(p-1)+r}\big((x^p - x)h\big)} \\
    &= \sum_{k\ge 1} S_p(x^p-x) S_{k(p-1)+r-p}(h) + \sum_{k\ge 0} S_1(x^p-x) S_{k(p-1)+r-1}(h) \\
    &= \sum_{k\ge 1} S_{k(p-1)+r-p}(h) - \sum_{k\ge 0} S_{k(p-1)+r-1}(h) = 0.
    \end{split}
    \]
    So $\sum_{k\ge 0} S_{k(p-1)+r}(q_\lambda) = \sum_{k\ge 0} S_{k(p-1)+r}(q'_{\lambda})$.
\end{proof}

\begin{remark}
    One more thing can be observed (but will not be needed): Given any LRS $(a_n)_n$ over a field of characteristic $p > 0$, there exists some power $p^k$ such that the subsequences $(a_{n p^k +r})_{n}$ are representable by an EPS (for every $r$ and sufficiently large $n$).
    For instance, while the sequence $(\overline{T_n})_n$ in \cref{exm:char2-not-exppoly} is not $2$-periodic, it is $4$-periodic, and splitting it into four subsequences modulo $4$, each subsequence is constant, and hence obviously an EPS.

    In general, the periodicity appears because $n \mapsto \overline{\binom{n+j-1}{j-1}}$ is still periodic with period $p^k$ for sufficiently large $k$. This can be seen as a consequence of Lucas's theorem for expressing binomial coefficients modulo $p$, but is also easy to prove directly.
\end{remark}

\section{Additional Material} \label{sec:additional-proofs}

\subsection{Translating CCRA into Weighted Automata}
\label{ssec:ccra-to-wa}
\begin{lemma}\label{lemma:copyless_to_WA}
For every CCRA $\C$ there is a weighted automaton $\W$, of size exponential in the size of $\C$, such that $\C$ and $\W$ are equivalent.
\end{lemma}

\begin{proof}
Fix a CCRA $\C = (Q, q_0, d, \Sigma, \delta, \mu, \nu)$.
We use the fact that linear CRA are equivalent to weighted automata and define an equivalent linear CRA $\C'= (Q, q_0, d', \Sigma, \delta', \mu', \nu')$ (the states remain the same). We note that $\C'$ need not be copyless. The new dimension is $d' = 2^d$ with the following intuition. If the variables $x_1$, \dots,~$x_d$ represent the registers of $\C$, then the registers in $\C'$ correspond to all possible square-free monomials, i.e., monomials of the form $\prod_{i \in I} x_i$ for every $I \subseteq \set{1,\ldots,d}$ (note that there are $2^d$ square-free monomials).

The automaton $\C'$ mimics the behaviour of $\C$ with the intuitive meaning that registers (monomials) in $\C'$ store the corresponding product of registers (variables) in $\C$. We observe that this can be maintained as an invariant. Recall that $x_1,\ldots x_d$ are registers in $\C$. Let $p_{q,a} = (p^1,\ldots,p^d) \in \Poly^d$ be a copyless polynomial mapping occurring in one of the transitions of $\C$. We show how to update accordingly registers in $\C'$. Let $m$ be a register in $\C'$ corresponding to a monomial $\prod_{i \in I} x_i$. We define the new value of $m$ as $\prod_{i \in I} p_i$, written as a sum of monomials. We observe that since $p$ is copyless, all monomials in $\prod_{i \in I} p_i$ are square-free. Thus we have updated $m$ as a linear combination of previous monomials.

It remains to define $\mu'$ so that the invariant also holds in the first step; and $\nu'$ as the expanstion of $\nu$ into a sum of monomials.
\end{proof}

\subsection{Proof of \autoref{l:polyamb_upper_triangular}}

\polyambuppertriangular*
\begin{proof}
The \emph{transition graph} of a word $w$ is the unlabeled directed graph, possibly containing loops, on the vertex set $\{1,...,d\}$ in which there is an edge $i \to j$ if and only if $M(w)[i,j] \ne 0$.
Let us consider the transition graph $G$ corresponding to $M(w)$ and how it relates to the transition graph $H$ corresponding to $M( w^{d!} )$.
Any edge of $H$ comes from a directed walk of length $d!$ in $G$.

We first show that the only closed directed walks in $H$ are loops.
Indeed, suppose to the contrary that $H$ contains a closed directed cycle of length at least two.
Then $H$ contains a directed cycle\footnote{A closed directed walk $i_0 \to i_1 \to \dots \to i_l$ with $i_0=i_l$ and $i_j \ne i_j'$ unless $\{j,j'\}=\{0,l\}$} of length $l \ge 2$ from some vertex $i$ to itself.
This cycle arises from a directed walk $C$ from $i$ to itself in $G$ of length $l d!$.
In particular, since there exists a directed walk in $G$ from $i$ to itself, there exists a directed cycle $D_0$ in $G$ that is based at $i$.
The length $k$ of $D_0$ is of course at most $d$, and hence divides $d!$.
Let $D \coloneqq D_0^{ld!/k}$ denote the $ld!/k$-fold repetition of $D_0$.
Then $D$ and $C$ are two distinct directed walks in $G$ of length $ld!$: after $d!$ steps, the walk $D$ will be at $i$, but $C$ will not.
This leads to a contradiction with polynomial ambiguity: the word $w^{nld!}$ for $n \ge 1$, gives rise to at least $2^n$ directed walks from $i$ to itself in $G$, because in each repetition of $w^{ld!}$ we can choose to either follow $C$ or $D$.
Since $\A$ is trim, the state $i$ lies on an accepting run, and hence there exist words $u$,~$v$ such that $u (w^{nld!})v$ has at least $2^n$ accepting runs.

Now, since the only closed directed walks in $H$ are loops, we can define a total order $\preccurlyeq$ on $\{1,...,d\}$ such that there is no directed walk in $H$ from $i$ to $j$ if $j \prec i$. 
Permuting the standard basis vectors correspondingly, which means conjugating the matrix $M(w^{d!})$ by a permutation matrix $P$, we find that $P M(w^{d!})P^{-1}$ is upper triangular.

In particular, the eigenvalues of $M(w^{d!})$ are precisely the diagonal entries.
The $i$-th diagonal entry is the sum of the weights of all directed walks from $i$ to itself in $\A$ that are labeled by $w^{d!}$.
If there were two such walks, that would directly contradict polynomial ambiguity.
This means that the diagonal entries of $M(w^{d!})$ come from at most one walk, meaning they are products of transition weights of the automaton.
\end{proof}

\subsection{Proof of \autoref{lemma:coeff_sum}}

\coeffsum*
\begin{proof}
Since $\chr K= 0$, each EPS has a unique exponential polynomial representing it (\cref{subsec:appendix-exppoly}), allowing us to prove the statement for a specific exponential polynomial representing the EPS.
Now notice that all the sequences in \cref{d:generable} have this property. 
We need to show that the property is preserved under sums and products.
Let $(a_n)_n$,~$(a_n')_n$ be $R$-generable EPS, and let $q$,~$q'$ be the exponential polynomials representing them, respectively.
Then $q+q'$ represents $(a_n+a_n')_n$, so the claim is holds for sums.
The product $qq'$ represents $(a_n a_n')_n$.
By \cref{l:kcoeffsum-product}, all the sums of $k$-degree coefficients of $qq'$ are in $R$.
\end{proof}

\subsection{Proof in \autoref{example:maybe_unrecognisable}}

We prove that $f$ is an $\tfrac{1}{2}\Z$-generable EPS.
To see this, let $u = u_1u_2 \hdots u_r$, $w = w_1w_2 \hdots w_t$, $v=v_1v_2 \hdots v_l$. 
Then (with some constants $C$, $D \in \N$, independent of $n$)
\[
\begin{split}
f(uw^nv) &= \sum_{j=1}^r j u_j + \sum_{k=0}^{n-1} \sum_{j=1}^t (r+kt+j) w_j + \sum_{j=1}^l (r+nt+j) v_j \\
&= Cn + D + \sum_{j=1}^t \Big(\sum_{k=0}^{n-1} (r+kt+j)\Big) w_j \\
&= Cn + D + \sum_{j=1}^t \Big((r+j)n + \frac{(n-1)n}{2}t\Big) w_j.
\end{split}
\]
We see that $n \mapsto f(uw^nv)$ is a polynomial function, with coefficients in $A=\tfrac{1}{2}\Z$.
Hence, the LRS $(f(uw^nv))_n$ is an $A$-generable EPS.

\subsection{Proof of \autoref{lem:formula_to_poly}}

\formulatopoly*
\begin{proof}
By induction on the formula size. For the base case $k=1$ and then given $\varphi = x_1$ we define $p = x_1$. Otherwise, suppose we have polynomials $p$, $p_1$ and $p_2$ corresponding to some formulas $\varphi$, $\varphi_1$ and $\varphi_2$, respectively. We build polynomials as follows.
\begin{itemize}
 \item For the formula $\neg \varphi$ we define the polynomial $1 - p$.
 \item For the formula $\varphi_1 \land \varphi_2$ we define the polynomial $p_1 \cdot p_2$.
\end{itemize}
It is easy to see that the construction preserves the properties of the lemma. In particular the base case formulas are trivially $1$-copyless; and by induction every subformula $\psi$ is $|\psi|$-copyless. Moreover, every Boolean formula can be build from $\neg$ and $\land$, and the final polynomials have polynomial size in the size of the input formula.
\end{proof}

\subsection{Additional Example in \autoref{subsection:examples}}

Another promising approach to showing the non-sufficiency of the conditions in \cref{theorem:CCRA}, is to show that the class of functions recognised by CCRA is not closed under reversal.
For tropical semirings this is known~\cite{MazowieckiR19}.

\begin{conjecture}
    For $\card{\Sigma} \ge 2$, there exists a function $f\colon \Sigma^* \to \Q$ that is recognisable by a $\Q$-CCRA, but for which the reversal $w \mapsto f(w^r)$, with $w^r$ denoting the reversal of $w$, is not recognisable by a $\Q$-CCRA.
\end{conjecture}

Again, there is a promising candidate for which we are currently unable to prove that the reverse is not recognisable.

\begin{example}
Consider the following function $f\colon \{0, 1\}^* \to \Q$.
\[
\underbrace{0 \hdots 0}_{m_1} \underbrace{1 \hdots 1}_{k_1} \underbrace{0 \hdots 0}_{m_2} \underbrace{1 \hdots 1}_{k_2} \hdots \underbrace{0 \hdots 0}_{m_t} \underbrace{1 \hdots 1}_{k_t} \mapsto ( \hdots ((m_1+k_1)m_2+k_2) \hdots )m_t + k_t
\]
It is easy to evaluate $f(w)$ with a CCRA --- the bracketing in the formula above can be interpreted as a recipe for doing so. However, there is no obvious way to recognise $f^r(w) \coloneqq f(w^r)$.

\end{example}

\subsection{\autoref{theorem:CCRA} in Positive Characteristic.}
\label{ssec:ccra-positive-char}

In this subsection we discuss how the second property in \autoref{theorem:CCRA} changes when the field $K$ has positive characteristic $p > 0$.

First, there is the issue that the exponential polynomial $q$ representing $h(n)$ is not unique.
As discussed in Appendix~\ref{subsec:appendix-exppoly}, this can be overcome by choosing $q$ to be of minimal degree.
Then, the polynomials $p_i$ have degree at most $p-1$, and the second property of \cref{theorem:CCRA} holds, that is, we have $S_k(h) \in R$ for all $k$.

However, without explicit control over $m$, the claim is vacuous in this context: suppose the first claim of \autoref{theorem:CCRA} holds with $m=m'p$. (The proof of \autoref{theorem:CCRA} shows that $m$ can always be replaced by any multiple, so this situation is not a restriction.)

Then $(g(m'n+m'))_n$ is an $R$-generable EPS, and hence so is $(g(m'n+m))_n$ (using $m=m'p$).
Let $(g(m'n+m))_n$ be represented by the exponential polynomial $q=\sum_{i=1}^t p_i(x) \lambda_i^x$.
Then $p_i(np)=p_i(0)$ is constant for all $i$ and $n \ge 0$, so
\[
(h(n))_n \coloneqq g(mn+m)=g(m'pn+m)=\sum_{i=1}^t p_i(0) \lambda_i^{pn}
\]
is represented by the exponential polynomial $q'=\sum_{i=1}^t p_i(0) \lambda_i^{px}$ with constant coefficients.
Now $S_0(q') \in R$ just says $\sum_{i=1}^t p_i(0) = h(0) \in R$, which is trivially true.
Further, $S_k(q')=0$ for all $k \ge 1$ also holds trivially.

To obtain a non-vacous claim, we must therefore control $m$ more explicitly.
The following version of \autoref{theorem:CCRA} is obtained using a closer inspection of the proof.

\begin{theorem}
\label{theorem:CCRA-pos-char}
If $R \subseteq K$ is a subsemiring and $f\colon \Sigma^* \to K$ is recognised by an $R$-CCRA with $r$ registers and $s$ states. Let $m \coloneqq (4r+2)!s!$. Then
\begin{itemize}
\item for every $g \in \PSF(f)$, the sequence $(h(n))_n = (g(m(n+1)))_n$ is an $R$-generable EPS, and
\item $q$ is the exponential polynomial of minimal degree representing $h$, then for every $k \in \N$ the sum of $k$-degree coefficients $S_k(q)$ is in $R$.
\end{itemize}
\end{theorem}

In particular, if $p$ is sufficiently large compared to $r$ and $s$, this still yields nontrivial restrictions on the coefficients of $q$.
However, the dependence on the number of registers and states makes this much less useful than the corresponding statement in characteristic $0$.

In positive characteristic, there are multiple exponential polynomials representing $h$, and one must be careful to indeed choose the one of minimal degree (in which all the occurring polynomials have degree at most $p-1$). 
That is illustrated in the next example. See also Appendix~\ref{subsec:appendix-exppoly}.

\begin{figure}
\centering
\begin{minipage}[t]{.4\textwidth}
    \captionsetup{width=0.9\linewidth}
    \centering

    \begin{tikzpicture}
    \node[state, right = 2cm of pout] (q0) {};
    \node[left = 1.5cm of q0] (qin) {};
    \node[right = 1.5cm of q0] (qout) {};
        
    \path
    (qin) edge[->] node[above] {$x, y\coloneqq \overline{0}$} node[below] {$z \coloneqq \overline{1}$} (q0)
    (q0) edge[->] [loop above] node[above, align=center] {$x\coloneqq x+ \overline{1}$, $y \coloneqq y+\overline{1}$,\\ $z \coloneqq z+\overline{1}$} (q0)
    (q0) edge[->] node[above] {$xyz+\overline{1}$} (qout);
    \end{tikzpicture}
            
    \captionof{figure}{The CCRA over $\bF_3$ illustrating coefficient sums in positive characteristic (\cref{exm:poschar-coeffsum}).}
    \label{fig:ccra_f3}
\end{minipage}
\end{figure}
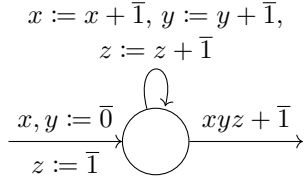
       
\begin{example} \label{exm:poschar-coeffsum}
    Consider the single-letter CCRA $\A$ in \cref{fig:ccra_f3} over $\bF_3=\{\overline{0}, \overline{1}, \overline{2}\}$.
    After reading $a^n$, the registers hold the values $(\overline{n}, \overline{n}, \overline{n} + 1)$.
    Hence, $\A(a^n) = \overline{n}^2(\overline{n}+\overline{1})+\overline{1} = \overline{n}^3 + \overline{n}^2 + \overline{1}$.
    However, the polynomial $x^3 + x^2 + 1$ is not the minimal degree polynomial representing $g(n)=\overline{n}^3 + \overline{n}^2 + 1$.
    Since $\overline{n}^3 + \overline{n}^2 + 1 = \overline{n}^2 + \overline{n} + 1$, instead $q=x^2+x+1$ is the minimal-degree representative.
    We have $S_0(q) = \overline{1}$, $S_1(q) = \overline{1}$, $S_2(q) = \overline{1}$ and $S_k(q) = \overline{0}$ for all other $k \in \N$.
\end{example}

The proof of \autoref{theorem:CCRA-pos-char} proceeds in the same way as that of \autoref{theorem:CCRA}.
We just need to replace \autoref{lemma:coeff_sum} by the following lemma.

\begin{lemma} \label{l:coeff-sums-ccra-poschar}
If $R$ is a subsemiring, $(a_n)_n$ is an $R$-generable EPS and $q$ is the exponential polynomial of minimal degree representing $(a_n)_n$, the sum of $k$-degree coefficients of $q$ is in $R$.
\end{lemma}

\begin{proof}
If $\chr K =0$, this is just \cref{lemma:coeff_sum}. Suppose $\chr K = p > 0$.

First notice that all the sequences in \cref{d:generable} have this property. 
We need to show that the property is preserved under sums and products.
Let $(a_n)_n$,~$(a_n')_n$ be $R$-generable EPS, and let $q$,~$q'$ be the exponential polynomials of minimal degree representing them, respectively.
Then $q+q'$ represents $(a_n+a_n')_n$ and is of minimal degree (since this just means $\deg(q+q') < p$).

Now $qq'$ represents $(a_n a_n')_n$.
By \cref{l:kcoeffsum-product}, all the sums of $k$-degree coefficients of $qq'$ are in $R$.
However, the product $qq'$ may not be the minimal degree exponential polynomial representing $(a_n a_n')_n$ (we only have $\deg(qq') \le 2p-2$, recall \cref{exm:poschar-coeffsum}).
Let $q''$ be the exponential polynomial of minimal degree representing $(a_n a_n')_n$.
Now \cref{l:charp-coeff-sums} implies $S_0(q'')=S_0(qq') \in R$ and $S_r(q'') = S_r(qq') + S_{r+p-1}(qq') \in R$ for all $r \in \{1, \dots, p-1\}$.
For $k \ge p$, we have $S_k(q'')=0$.
So again, it holds that $S_k(q'') \in R$ for all $k$.
\end{proof}

\end{document}